\documentclass[12pt,reqno]{amsart}

\usepackage{amsmath,amsfonts,amsthm,amssymb,amsxtra}
\usepackage{hyperref}
\usepackage{color}
\usepackage{bbm} 
\usepackage{bm} 

\newtheorem{theorem}{Theorem}[section]
\newtheorem{proposition}[theorem]{Proposition}
\newtheorem{lemma}[theorem]{Lemma}

\theoremstyle{definition}

\newtheorem{definition}[theorem]{Definition}
\newtheorem{example}[theorem]{Example}

\theoremstyle{remark}

\newtheorem{remark}[theorem]{Remark}

\numberwithin{equation}{section}

\newcommand{\cb}{\mathcal{B}}

\newcommand{\C}{\mathbb{C}}

\renewcommand{\epsilon}{\varepsilon}
\newcommand{\F}{\mathcal{F}}
\newcommand{\loc}{{\rm loc}}

\newcommand{\mt}{\mathrm{MT}}
\newcommand{\N}{\mathbb{N}}

\renewcommand{\phi}{\varphi}
\newcommand{\R}{\mathbb{R}}

\newcommand{\cw}{\mathcal{W}}

\newcommand{\Z}{\mathbb{Z}}

\newcommand{\high}{\mathrm{high}}
\newcommand{\low}{\mathrm{low}}
\newcommand{\reg}{\mathrm{reg}}
\newcommand{\sing}{\mathrm{sing}}

\DeclareMathOperator{\dist}{dist}
\DeclareMathOperator{\Div}{div}

\DeclareMathOperator{\supp}{supp}
\DeclareMathOperator{\sgn}{sgn}

\def\bs{\mathbb{S}}

\def\cs{\mathcal{S}}
\def\cv{\mathcal{V}}

\def\rd{\mathrm{d}}

\def\dk{\mathrm{d}k}
\def\dr{\mathrm{d}r}

\def\dt{\mathrm{d}t}
\def\dx{\mathrm{d}x}

\newcommand{\me}[1]{\mathrm{e}^{#1}}
\newcommand{\one}{\mathbf{1}}

\newcommand{\e}{{\rm e}}
\newcommand\I{\mathrm{i}}
\DeclareMathOperator{\rk}{rank}

\makeatletter
\newcommand*{\rom}[1]{\expandafter\@slowromancap\romannumeral #1@}
\makeatother

\begin{document}

\title[Weak coupling limit]{Weak coupling limit for Schr\"odinger-type operators with degenerate kinetic energy for a large class of potentials}

 \author[J.-C. Cuenin]{Jean-Claude Cuenin}
 \address[Jean-Claude Cuenin]{Department of Mathematical Sciences, Loughborough University, Loughborough,
 Leicestershire, LE11 3TU United Kingdom}
 \email{J.Cuenin@lboro.ac.uk}

 \author[K. Merz]{Konstantin Merz}
 \address[Konstantin Merz]{Institut f\"ur Analysis und Algebra, Technische Universit\"at Braunschweig, Universit\"atsplatz 2, 38106 Braunschweig, Germany}
 \email{k.merz@tu-bs.de}


\date{March 04, 2021}

\begin{abstract}
  We improve results by Frank, Hainzl, Naboko, and Seiringer
  \cite{MR2365659} and Hainzl and Seiringer \cite{MR2643024}
  on the weak coupling limit of eigenvalues for
  Schr\"odinger-type operators whose kinetic energy vanishes
  on a codimension one submanifold. The main technical
  innovation that allows us to go beyond the potentials
  considered in \cite{MR2365659,MR2643024} is the use of the
  Tomas--Stein theorem.
\end{abstract}

\maketitle
\section{Introduction and main results}

There has been recent interest in Schr\"odinger-type operators
of the form
\begin{align}
  \label{eq:defhlambda}
  H_{\lambda} = T(-\I\nabla)-\lambda V \quad \text{in}\ L^2(\R^d)\,,
\end{align}
where the kinetic energy $T(\xi)$ vanishes on a submanifold of
codimension one, $V$ is a real-valued potential, and $\lambda>0$
is a coupling constant. We are interested in the weak coupling
limit $\lambda\to 0$ for potentials that decay slowly in some
$L^p$ sense to be made precise. Operators of this type appear in
many areas of mathematical physics
\cite{MR0404846,MR1671985,MR772044,MR1970614,MR2303586,MR2466689,MR2250814,PhysRevB.77.184517,MR2365659,MR2450161,MR2410898,MR2643024,hoang2016quantitative,Gontier_2019}. 
The goal of \cite{MR2643024} was to generalize the results and
techniques of \cite{MR2365659} and \cite{PhysRevB.77.184517} to
a large class of kinetic energies. Our goal, complementary to
\cite{MR2643024}, is to relax the conditions on the potential.
To keep technicalities to a minimum, we state our result for
$T(-\I\nabla)=|\Delta+1|$. This was one of the main motivations
to study operators of the form \eqref{eq:defhlambda}, due to their
role in the BCS theory of superconductivity
\cite{MR2365659,MR2410898}. As in previous works
\cite{MR2365659,MR2643024} a key role is played by an operator
$\mathcal{V}_S$ on the unit sphere $S\subset\R^d$, whose convolution
kernel is given by the Fourier transform of $V$. The potentials we
consider here need not be in $L^1(\R^d)$, but $\mathcal{V}_S$ may be
defined as a norm limit of a regularized version (see Section
\ref{Section def V_S} for details). The potential $V$ is assumed to
belong to the amalgamated space $\ell^{\frac{d+1}{2}}L^{\frac{d}{2}}$,
where the first space measures global (average) decay and the
second measures local regularity (see \eqref{potential class def.}).
We note that
$L^{\frac{d+1}{2}}\cup L^{\frac d2}\subseteq \ell^{\frac{d+1}{2}}L^{\frac d2}$
by Jensen's inequality.
\begin{theorem}
  \label{thm. asymptotics}
  Let $d\geq 3$ and $H_{\lambda} =|\Delta+1|-\lambda V$. 
  If $V\in\ell^{\frac{d+1}{2}}L^{\frac{d}{2}}$, then for every eigenvalue
  $a_S^j>0$ of $\mathcal{V}_S$ in \eqref{eq:lswbs2}, counting
  multiplicity, and every $\lambda>0$, there is an eigenvalue
  $-e_j(\lambda)<0$ of $H_{\lambda}$ with weak coupling limit
  \begin{align}
    \label{eq. weak coupling limit}
    e_j(\lambda) = \exp\left(-\frac{1}{\lambda a_S^j}(1+o(1))\right)
    \quad \text{as}\ \lambda \to 0.
  \end{align}
\end{theorem}
For simplicity we stated the result for $d\geq 3$, but it easily
transpires from the proof that it also holds in $d=2$ for
$V\in\ell^{\frac{d+1}{2}}L^{1+\epsilon}$ for arbitrary $\epsilon>0$.
All other possible negative eigenvalues (not corresponding to
$\mathcal{V}_S$) satisfy $e_j(\lambda)\leq \exp(-c/\lambda^2)$.
The statement in \cite{MR2643024} about the convergence of
eigenfunctions also holds for the potentials considered here. Since
the proofs are completely analogous we will not discuss them.

In previous works \cite{MR2365659,MR2643024} it was assumed that
$V\in L^{1}(\R^d)\cap L^{\frac{d}{2}}(\R^d)$. Our main contribution
is to remove the $L^1$ assumption, allowing for potentials with
slower decay. 
The main new idea is to use the Tomas--Stein theorem (see Subsection
\ref{subsection def. V_s} and \eqref{eq:tsorg}, \eqref{eq:tsexplicit}).
In view of its sharpness, our result is optimal in the sense that the
exponent $(d+1)/2$ in our class of admissible potentials cannot be
increased, unless one imposes further (symmetry) restrictions on $V$,
see also the discussion below. Moreover, the use of amalgamated spaces
allows us to relax the global regularity to the local condition
$V\in L_{\loc}^{\frac{d}{2}}(\R^d)$ which just suffices to guarantee
that $H_{\lambda}$ is self-adjoint.

The idea of applying the Tomas--Stein theorem and related results such
as \cite{MR894584} to problems of mathematical physics is not new, see,
e.g., \cite{MR2219246} and \cite{MR2820160}. The validity of the
Tomas--Stein theorem crucially depends on the curvature of the
underlying manifold. A slight modification of our proof (see, e.g.,
\cite{MR3608659,MR3942229}) shows that the result of Theorem
\ref{thm. asymptotics} continues to hold for general Schr\"odinger-type
operators (with a suitable modification of the local regularity
assumption) of the form \eqref{eq:defhlambda} as long as the Fermi
surface $S=\{\xi\in\R^d:\ T(\xi)=0\}$ is smooth and has everywhere
non-vanishing Gaussian curvature. For example, if $T$ is elliptic at
infinity of order $2d/(d+1)\leq s<d$, then the assumption on the
potential becomes $V\in\ell^{\frac{d+1}{2}}L^{\frac{d}{s}}$. This is
outlined in Theorem \ref{asymptoticsgen} and improves
\cite[Theorem 2.1]{MR2643024}. The moment-type condition on the
potential in that theorem is unnecessary, regardless of whether the
kinetic energy is radial or not.
A straightforward generalization to the case where $S$ has at least
$k$ non-vanishing principal curvatures can be obtained from the
results of \cite{MR620265,MR3942229}. In that case the global decay
assumption has to be strengthened to
$V\in\ell^{\frac{k+2}{2}}L^{\frac{d}{s}}$. Sharp restriction theorems for
surfaces with degenerate curvature are available in the
three-dimensional case \cite{MR3524103}. 

Based on the results of \cite{MR1479544,MR3713021,Vega1992}, if the
potential $V$ is radial, one might be able to relax the assumption in
Theorem \ref{thm. asymptotics} to $V\in\ell^{d}L^{\frac{d}{2}}$.
This naive belief is supported by the discussion in Appendix \ref{a:mt},
see especially Theorem \ref{asymptoticsradial} where we generalize
Theorem \ref{thm. asymptotics} to spherically symmetric potentials with
almost $L^d$ decay.

For long-range potentials the weak coupling limit
\eqref{eq. weak coupling limit} does not hold in general.
Gontier, Hainzl, and Lewin \cite{Gontier_2019} showed
$\exp(-C_1/\sqrt{\lambda})\leq e_1(\lambda)\leq\exp(-C_2/\sqrt{\lambda})$
for the Coulomb potential $V=|x|^{-1}$ in $d=3$.

The key estimate \eqref{Key estimate} is a consequence of the
Tomas--Stein theorem. The remainder of the proof is standard first
order perturbation theory that is done in exactly the same way as
in \cite{MR2365659,MR2643024}. In a similar manner -- again following
\cite{PhysRevB.77.184517,MR2643024} -- we will carry out higher order
perturbation theory in Subsection \ref{ss:higherorders} and show how
one may in principle obtain any order in the asymptotic expansion of
$e_j(\lambda)$ at the cost of restricting the class of admissible
potentials. For instance, our methods allow us to derive the second
order for $V\in L^{\frac{d+1}{2}-\epsilon}$ and some $\epsilon\in(0,1/2]$.
For $V\in L^1\cap L^{d/2}$ this was first carried out in
\cite{PhysRevB.77.184517,MR2643024}.
Furthermore, we will give an alternative proof for the existence of
eigenvalues of $H_\lambda$ based on Riesz projections in Subsection
\ref{Section alternative proof}. This approach allows us to handle
complex-valued potentials\footnote{In this case, a transformation of
  statements about non-self-adjoint operators into those about a
  self-adjoint operator as in the proof of Theorem
  \ref{thm. asymptotics} seems impossible.} on the same footing as
real-valued ones. The former play a role, e.g., in the theory of
resonances, but are also of independent interest.

We use the following notations: For two non-negative numbers $a,b$
the statement $a\lesssim b$ means that $a\leq C b$ for some universal
constant $C$. If the estimate depends on a parameter $\tau$, we indicate
this by writing $a\lesssim_{\tau} b$. The dependence on the dimension
$d$ is always suppressed. We will assume throughout the article that
the (asymptotic) scales $e$ and $\lambda$ are positive, sufficiently
small, and that $\lambda\ln(1/e)$ remains uniformly bounded from above
and below. The symbol $o(1)$ stands for a constant that tends to zero
as $\lambda$ (or equivalently $e$) tends to zero. We set
$\langle\nabla\rangle=(\mathbf{1}-\Delta)^{1/2}$.


\section{Preliminaries}
\label{s:prelims}

\subsection{Potential class}

Let $\{Q_s\}_{s\in\Z^d}$ be a collection of axis-parallel unit cubes
such that $\R^d=\bigcup_{s}Q_s$. We then define the norm
\begin{align}
  \label{potential class def.}
  \|V\|_{\ell^{\frac{d+1}{2}} L^{\frac{d}{2}}}
  := \left[\sum_s\|V\|_{L^{\frac{d}{2}}(Q_s)}^{\frac{d+1}{2}}\right]^{\frac{2}{d+1}}.
\end{align}
The exponent $(d+1)/2$ is natural (cf. \cite{MR2038194,MR2252331})
in view of the Tomas--Stein theorem. This is the assertion that the
Fourier transforms of $L^p(\R^d)$ functions indeed belong to $L^2(S)$
whenever $p\in[1,\kappa]$ where $\kappa=2(d+1)/(d+3)$ denotes the
``Tomas--Stein exponent''. We discuss this theorem and a certain
extension thereof in more detail in the next subsection. Observe that
$1/\kappa-1/\kappa'=2/(d+1)$. The following lemma is a straightforward
generalization of \cite[Lemma 6.1]{MR2219246}.

\begin{lemma}\label{lemma Ionescu--Schlag}
  Let $s\geq2d/(d+1)$ and $V\in \ell^{\frac{d+1}{2}}L^{\frac{d}{s}}$.
  Then
  \begin{align}
    \||V|^{1/2}\langle\nabla\rangle^{-\left(\frac{s}{2}-\frac{d}{d+1}\right)}\phi\|_{L^2}
    \lesssim \|V\|^{1/2}_{\ell^{\frac{d+1}{2}}L^{\frac{d}{s}}}\|\phi\|_{L^{\kappa'}}.
  \end{align}
\end{lemma}

\begin{proof}
  We abbreviate $\alpha=s/2-d/(d+1)\geq0$ and first note that, by duality,
  the assertion is equivalent to
  $$
  \|\langle\nabla\rangle^{-\alpha}|V|^{1/2}\phi\|_{L^\kappa}
  \lesssim \|V\|^{1/2}_{\ell^{\frac{d+1}{2}}L^{\frac{d}{s}}}\|\phi\|_{L^2}\,.
  $$
  If $\alpha=0$, the claim follows from H\"older's inequality,
  $d/s=(d+1)/2$ in this case, and $\ell^p L^p= L^p$ for all $p\in[1,\infty]$.
  On the other hand, if $\alpha\geq d$ we use the fact that
  $\langle\nabla\rangle^{-\gamma}$ is $L^p$ bounded for
  all $p\in(1,\infty)$ and $\gamma\geq0$ (by the H\"ormander--Mihlin
  multiplier theorem, cf. \cite[Theorem 6.2.7]{MR3243734}).
  Thus we shall show
  $$
  \||V|^{1/2}\langle\nabla\rangle^{-\alpha}\phi\|_{L^2}
  \lesssim \|V\|^{1/2}_{\ell^{\frac{d+1}{2}}L^{\frac{d}{s}}}\|\phi\|_{L^{\kappa'}}
  $$
  for $\alpha=s/2-d/(d+1)$ with $s\geq2d/(d+1)$ such that $\alpha\in(0,d)$.
  Let $\{Q_s\}_{s\in\Z^d}$ be the above family of axis-parallel unit
  cubes tiling $\R^d$, i.e., for $s\in\Z^d$ let
  $Q_s=\{x\in\R^d:\, \max_{j=1,...,d}|x_j-s_j|\leq1/2\}$. Next, recall that
  for $\alpha\in(0,d)$, we have for any $N\in\N_0$
  \begin{align*}
    |\langle\nabla\rangle^{-\alpha}\phi(x)| \lesssim_{\alpha,N} |\phi|\ast W_\alpha(x)
  \end{align*}
  where
  \begin{align}
    \label{eq:defwalpha}
    W_\alpha(x) = |x|^{-(d-\alpha)} \one_{\{|x|\leq1\}} + |x|^{-N}\one_{\{|x|\geq1\}}\,.
  \end{align}
  (For a proof of these facts, see, e.g., \cite[p. 132]{Stein1970}.)
  Abbreviating further $q_0=d/s$, we obtain
  \begin{align*}
    \||V|^{1/2}\langle\nabla\rangle^{-\alpha}\phi\|_{L^2}^2
    & \lesssim_{\alpha,N} \sum_{s\in\Z^d} \int_{Q_s}|V(x)| [(|\phi|\ast W_\alpha)(x)]^2\,dx\\
    & \leq \sum_{s\in\Z^d} \|V\|_{L^{q_0}(Q_s)}\cdot \||\phi|\ast W_\alpha\|_{L^{2q_0'}(Q_s)}^2\\
    & \leq \sum_{s\in\Z^d} \|V\|_{L^{q_0}(Q_s)}\left[\sum_{s'\in\Z^d}\|(\one_{Q_{s'}}|\phi|)\ast W_\alpha\|_{L^{2q_0'}(Q_s)}\right]^2\\
    & \lesssim \sum_{s\in\Z^d} \|V\|_{L^{q_0}(Q_s)}\left[\sum_{s'\in\Z^d}\|\phi\|_{L^{\kappa'}(Q_{s'})}(1+|s-s'|)^{-N}\right]^2\\
    & \lesssim_N \left[\sum_{s\in\Z^d}\|V\|_{L^{q_0}(Q_s)}^{(d+1)/2}\right]^{2/(d+1)}\|\phi\|_{L^{\kappa'}}^2
  \end{align*}
  where we used H\"older's inequality in the second line, the
  Hardy--Littlewood--Sobolev inequality in the penultimate line,
  and H\"older's and Young's inequality in the last line.
  This concludes the proof.
\end{proof}

\subsection{Definition of $\mathcal{V}_S$}
\label{Section def V_S}\label{subsection def. V_s}

As observed in \cite{MR1970614}, the weak coupling limit of
$e_j(\lambda)$ is determined by the behavior of the potential
on the zero energy surface of the kinetic energy, i.e., on
the unit sphere $S$. We denote the Lebesgue measure on $S$
by $\rd\omega$. For $V\in L^1(\R^d)$ we consider the
self-adjoint operator $\mathcal{V}_S:L^2(S)\to L^2(S)$,
defined by
\begin{align}
  \label{eq:lswbs}
  (\mathcal{V}_Su)(\xi)
  = \int_S \widehat{V}(\xi-\eta) u(\eta)\,\rd\omega(\eta),
  \quad u\in L^2(S),
\end{align}
see, e.g., \cite[Formula (2.2)]{MR2365659}. Here we have
absorbed the prefactors in the definition of the Fourier
transform, i.e., we use the convention
\begin{align*}
  \widehat{V}(\xi)=\int_{\R^d}\e^{-2\pi \I x\cdot\xi}V(x)\rd x.
\end{align*} 
Our definition of $\mathcal{V}_S$ differs from that of
\cite{MR2365659,MR2643024} by a factor of $2$; this is
reflected in the formula \eqref{eq. weak coupling limit}.
Since $V\in L^1(\R^d)$, its Fourier transform is a bounded
continuous function by the Riemann--Lebesgue lemma and is
therefore defined pointwise. The Tomas--Stein theorem allows
us to extend the definition of $\mathcal{V}_S$ to a larger
potential class. To this end we observe that the operator
in \eqref{eq:lswbs} can be written as
\begin{align}
  \label{eq:lswbs2}
  \cv_S = \mathcal{F}_{S}V\mathcal{F}_{S}^*\,,
\end{align}
where
$\mathcal{F}_{S}:\cs(\R^d)\to L^2(S)$, $\phi\mapsto\widehat{\phi}|_S$
is the Fourier restriction operator (here $\mathcal{S}$ is
the Schwartz space on $\R^d$). Its adjoint, the Fourier
extension operator $\mathcal{F}_S^*:L^2(S)\to \cs'(\R^d)$,
is given by
\begin{align}
  (\mathcal{F}_S^* u)(x) = \int_S u(\xi) \me{2\pi ix\cdot\xi} \,\rd\omega(\xi)\,.
\end{align}
A fundamental question in harmonic analysis is to find
optimal sufficient conditions for $\kappa$ such that
$\mathcal{F}_S$ is an $L^{\kappa}\to L^q$ bounded operator.
By the Hausdorff--Young inequality, the case $\kappa=1$ is
trivial. On the other hand, the Knapp example (see, e.g.,
\cite[p. 387-388]{Stein1993}) and the decay of the Fourier
transform of the surface measure \cite{Herz1962} show that
$\kappa<2d/(d+1)$ and $(d+1)/\kappa'\leq(d-1)/q$ are necessary
conditions.
The content of the Tomas--Stein theorem (unpublished, but
see, e.g., Stein \cite[Theorem 3]{Stein1986} and Tomas
\cite{Tomas1975}) is that, for $q=2$, these conditions are
indeed also sufficient. Concretely, the estimate
\begin{align}
  \label{eq:tsorg}
  \|\F_S\phi\|_{L^2(S)} \lesssim \|\phi\|_{L^p(\R^d)}
  \,, \quad p\in[1,\kappa]\,, \quad \kappa=2(d+1)/(d+3)
\end{align}
holds for all $d\geq2$, whenever $S$ is a smooth and compact
hypersurface with everywhere non-zero Gaussian curvature.
In particular, this estimate is applicable to the Fermi
surfaces that we consider later in Subsection
\ref{generalkinen}. Moreover, by H\"older's inequality it
follows that $|V|^{1/2}\mathcal{F}_S^*$ is an
$L^2(S)\to L^2(\R^d)$ bounded operator, whenever
$V\in L^q(\R^d)$ and $q\in[1,(d+1)/2]$. In this case,
$\cv_S$ is of course $L^2(S)$ bounded as well with
\begin{align}
  \label{eq:tsexplicit}
  \|\cv_S\| \lesssim \|V\|_{L^{q}}\,, \quad q\in[1,(d+1)/2]\,.
\end{align}
In the following, we will often refer to this estimate as the
Tomas--Stein theorem. Recently, Frank and Sabin
\cite[Theorem 2]{MR3730931} extended \eqref{eq:tsexplicit}
and showed
\begin{align}
  \label{eq:tsfsexplicitgen}
  \|W_1\F_S^*\F_SW_2\|_{\mathfrak{S}^{\frac{(d-1)q}{d-q}}}
  \lesssim_{q} \|W_1\|_{L^{2q}}\|W_2\|_{L^{2q}} \,, \quad W_1,W_2\in L^{2q}(\R^d)\,,\ q\in[1,(d+1)/2]
\end{align}
where $\mathfrak{S}^q(L^2)$ denotes the $q$-th Schatten
space over $L^2$. Observe that the Schatten exponent is
monotonously increasing in $q$. In particular, taking
$q=(d+1)/2$, $W_1=|V|^{1/2}$, and $W_2=V^{1/2}$ where
$V^{1/2}=|V|^{1/2}\sgn(V)$ with $\sgn(V(x))=1$ whenever
$V(x)=0$, shows that $\cv_S$ belongs to
$\mathfrak{S}^{d+1}(L^2(S))$ with
\begin{align}
  \label{eq:tsfsexplicit}
  \|\cv_S\|_{\mathfrak{S}^{d+1}} \lesssim \|V\|_{L^{(d+1)/2}}\,.
\end{align}

We will now extend the definition of \eqref{eq:lswbs2} to
incorporate potentials in the larger class
$\ell^{(d+1)/2}L^{d/2}$ that appears in our main result.
\begin{proposition}
  \label{proposition def. V_Sw}
  Let $V\in \ell^{\frac{d+1}{2}}L^{\frac{d}{2}}$. Then
  \eqref{eq:lswbs2} defines a bounded operator on $L^2(S)$.
  Moreover, if $(V_n)_n$ is a sequence of Schwartz functions
  converging to $V$ in $\ell^{\frac{d+1}{2}}L^{\frac{d}{2}}$
  and $\mathcal{V}_S^{(n)}$ are the corresponding operators
  in \eqref{eq:lswbs}, then $\mathcal{V}_S$ is the norm
  limit of the $\mathcal{V}_S^{(n)}$.
\end{proposition}

\begin{proof}
  We first assume that $V\in L^{\frac{d+1}{2}}(\R^d)$. It
  follows from the above discussion that $\mathcal{V}_S$ is
  the norm limit of the $\mathcal{V}_S^{(n)}$. To extend the
  definition to all $V\in \ell^{\frac{d+1}{2}}L^{\frac{d}{2}}$,
  we prove
  \begin{align}
    \label{enhanced Tomas--Stein}
    \|\mathcal{F}_{S}V\mathcal{F}_{S}^*\|
    \lesssim \|V\|_{\ell^{\frac{d+1}{2}}L^{\frac{d}{2}}}.
  \end{align}
  To this end we use the following observation.
  For $u\in L^2(S)$ and $\xi\in S$ we write
  \begin{align}
    \label{eq:trick}
    (\mathcal{V}_S^{(n)}u)(\xi)
    = \int_S (\widehat{V_n}\phi)(\xi-\eta) u(\eta)\,\rd \omega(\eta),
  \end{align}
  where $\phi$ is a bump function that equals $1$ in $B(0,2)$.
  This has the same effect as replacing $V_n$ by $\phi^{\vee}*V_n$.
  (Here,
  $\phi^\vee(x):=\int_{\R^d}\me{2\pi ix\cdot\xi}\phi(\xi)\,\rd\xi$
  denotes the inverse Fourier transform.) Since
  \eqref{enhanced Tomas--Stein} is equivalent to the bound
  \begin{align}
    \label{enhanced Tomas--Stein TT*}
    \|\sqrt{|V|}\mathcal{F}_{S}^*\mathcal{F}_{S}\sqrt{V}\|
    \lesssim \|V\|_{\ell^{\frac{d+1}{2}}L^{\frac{d}{2}}},
  \end{align}
  where $V^{1/2}=|V|^{1/2}\sgn(V)$ and $\sgn(V)$ is a unitary
  multiplication operator, we may assume without loss of generality
  that $V\geq 0$. Passing to a subsequence, we may also assume that
  $(V_n)_n$ converges to $V$ almost everywhere. By Fatou's lemma,
  for any $u\in L^2(S)$, 
  \begin{align}
    \label{eq:defvsaux}
    \begin{split}
      \langle \mathcal{F}_{S}^*u,V\mathcal{F}_{S}^*u\rangle
      & \leq \liminf_{n\to\infty}\langle \mathcal{F}_{S}^*u,V_n\mathcal{F}_{S}^*u\rangle
      \leq \liminf_{n\to\infty}\|(\phi^{\vee}\ast V_n)(F_S^*u)\|_{L^{\kappa}}\|u\|_2\\
      & \lesssim \|V\|_{\ell^{\frac{d+1}{2}}L^{\frac{d}{2}}}\|u\|_2^2\,,
    \end{split}
  \end{align}
  where the penultimate inequality follows from the Tomas--Stein
  theorem \eqref{eq:tsexplicit} and the last inequality from
  the bound
  \begin{align}
    \label{smoothing bound}
    \|(\phi^{\vee}\ast V)(F_S^*u)\|_{L^\kappa}
    \leq \|\phi^\vee\ast V\|_{L^{\frac{d+1}{2}}}\|F_S^*u\|_{L^{\kappa'}}
    \lesssim_{\phi} \|V\|_{\ell^{\frac{d+1}{2}}L^{\frac{d}{2}}}\|u\|_{L^2}
  \end{align}
  whose proof is similar to that of Lemma \ref{lemma Ionescu--Schlag}
  since the convolution kernel of $\phi^\vee$ is a Schwartz function, i.e.,
  in particular $|\phi^\vee(x)|\lesssim_N(1+|x|)^{-N}$ for any $N\in\N$.
  More precisely, for the same family $\{Q_s\}_{s\in\Z^d}$ of axis-parallel
  unit cubes tiling $\R^d$ that we used in the proof of Lemma
  \ref{lemma Ionescu--Schlag}, we have for any $N>0$,
  \begin{align}
    \label{eq:pfsmoothingbound}
    \begin{split}
      \|\phi^\vee\ast V\|_{\frac{d+1}{2}}^{\frac{d+1}{2}}
      & =  \|\sum_{s}\one_{Q_s}(\phi^\vee\ast V)\|_{\frac{d+1}{2}}^{\frac{d+1}{2}}
      = \sum_s \|\phi^\vee\ast (\sum_{s'}V\one_{Q_{s'}})\|_{L^{\frac{d+1}{2}}(Q_s)}^{\frac{d+1}{2}}\\
      & \leq \sum_s\left[\sum_{s'}\|\phi^\vee\ast(V\one_{Q_{s'}})\|_{L^{\frac{d+1}{2}}(Q_s)}\right]^{\frac{d+1}{2}}\\
      & \lesssim_N \sum_s\left[\sum_{s'}(1+|s-s'|)^{-N}\|V\|_{L^{\frac d2}(Q_{s'})}\right]^{\frac{d+1}{2}}
      \lesssim_N \|V\|_{\ell^{\frac{d+1}{2}}L^{\frac{d}{2}}}^{\frac{d+1}{2}}
    \end{split}
  \end{align}
  where we used Young's inequality in the last two estimates.
  This concludes the proof.
\end{proof}

\subsection{Compactness of $\mathcal{V}_S$}

We show that $\mathcal{V}_S$ belongs to a certain Schatten space
$\mathfrak{S}^{p}(L^2(S))$ and is thus a compact operator. In
particular, the spectrum of $\mathcal{V}_S$ is compact and countable
with accumulation point $0$. The nonzero elements are eigenvalues
of finite multiplicity. That $0$ is in the spectrum follows from
the fact that $L^2(S)$ is infinite-dimensional.

\begin{lemma}
  \label{lemma compactenss of V_S}
  Let $V\in \ell^{\frac{d+1}{2}}L^{\frac{d}{2}}$. Then
  $\mathcal{V}_S\in \mathfrak{S}^{d+1}(L^2(S))$ and
  \begin{align*}
    \|\mathcal{V}_S\|_{\mathfrak{S}^{d+1}}
    \lesssim \|V\|_{\ell^{\frac{d+1}{2}}L^{\frac{d}{2}}}.
  \end{align*}
\end{lemma}

\begin{proof}
  We recycle the proof of Proposition \ref{proposition def. V_Sw}
  and suppose $V\geq0$ without loss of generality again. We apply
  the Tomas--Stein theorem \eqref{eq:tsfsexplicit} for trace ideals
  with $V$ replaced by $\phi^\vee\ast V$ where $\phi$ is the same
  bump function as in that proof. Note that, by \eqref{eq:trick},
  this replacement does not affect the value of
  $\|\mathcal{V}_S\|_{\mathfrak{S}^{d+1}}$ since the eigenvalues
  remain the same. Thus, by \eqref{eq:pfsmoothingbound},
  $\|\mathcal{V}_S\|_{\mathfrak{S}^{d+1}}\lesssim \|\phi^\vee\ast V\|_{L^{(d+1)/2}} \lesssim \|V\|_{\ell^{(d+1)/2}L^{d/2}}$.
\end{proof}

\subsection{Birman--Schwinger operator}

As in \cite{MR2365659,MR2643024}, our proof is based on the
well-known Birman--Schwinger principle. This is the assertion
that, if 
\begin{align}
  \label{BS(e)}
  BS(e):=\sqrt{|V|}(T+e)^{-1}\sqrt{V}
\end{align}
with $e>0$, then 
\begin{align}
  \label{eq:bsprinc}
  -e\in \mathrm{spec}\left(H_\lambda\right)\iff 
    \frac{1}{\lambda}\in \mathrm{spec}\left(BS(e)\right).
\end{align}	
Here $\sqrt{V}:=\sgn(V)\sqrt{|V|}$ and $T=|\Delta+1|$. Thus,
\eqref{eq. weak coupling limit} would follow from
\begin{align}
  \label{To show}
  \ln(1/e)a_S^j(1+o(1))\in \mathrm{spec}(BS(e))
\end{align}
for every eigenvalue $a_S^j>0$ of $\mathcal{V}_S$. We note
that since $V$ and the symbol of $(T+e)^{-1}$ both vanish at
infinity, $BS(e)$ is a compact operator, see, e.g.,
\cite[Chapter 4]{MR2154153}. Moreover, we have the following
operator norm bound.

\begin{lemma}
  \label{lemma bound BS(e)}
  Let $V\in \ell^{\frac{d+1}{2}}L^{\frac{d}{2}}$. Then 
  \begin{align*}
    \|BS(e)\|\lesssim \ln(1/e)\|V\|_{\ell^{\frac{d+1}{2}}L^{\frac{d}{2}}}
  \end{align*}
  for all\, $e\in (0,1/2)$.
\end{lemma}

\begin{proof}
  The proof follows from \eqref{BShigh bound} and
  \eqref{BSlow bound} below.
\end{proof}

\section{Proof of Theorem \ref{thm. asymptotics}}
\label{proofmainresult}

\subsection{Outline of the proof}
We briefly sketch the strategy of the proof of \eqref{To show}.
We first split the Birman--Schwinger operator into a sum of high
and low energy pieces
\begin{align*}
  BS(e)=BS^{\rm low}(e)+BS^{\rm high}(e).
\end{align*}
More precisely, we fix $\chi\in C_c^{\infty}(\R^d)$ such that
$0\leq \chi\leq 1$ and $\chi\equiv 1$ on the unit ball. We also
fix $0<\tau<1$ and set
\begin{align*}
  BS^{\rm low}(e)=\sqrt{|V|}\chi(T/\tau)(T+e)^{-1}\sqrt{V}.
\end{align*}
As we will see in \eqref{BShigh bound}, the high energy piece
is harmless. The low energy piece is split further into a
singular and a regular part,
\begin{align*}
  BS^{\rm low}(e)=BS^{\rm low}_{\rm sing}(e)+BS^{\rm low}_{\rm reg}(e),
\end{align*} 
where the singular part is defined as
\begin{align}
  \label{BSlowsing def.}
  BS^{\rm low}_{\rm sing}(e)
  = \ln\left(1+\tau/e\right)\sqrt{|V|}\mathcal{F}_{S}^*\mathcal{F}_{S}\sqrt{V}.
\end{align}
Note that $\sqrt{|V|}\mathcal{F}_{S}^*\mathcal{F}_{S}\sqrt{V}$ is
isospectral to $\mathcal{V}_S$. As already mentioned in the
introduction and the previous section, Theorem
\ref{thm. asymptotics} would follow from standard perturbation
theory if we could show the key bound
\begin{align}
  \lambda\|BS^{\rm low}_{\rm reg}(e)\|&=o(1)\label{Key estimate}
\end{align}
for $V\in \ell^{\frac{d+1}{2}}L^{\frac{d}{2}}$, as long as
$\lambda\ln(1/e)$ remains uniformly bounded from above and below. 

\subsection{Bound for $BS^{\rm high}(e)$}

Here we prove that
\begin{align}
  \label{BShigh bound}
  \|BS^{\rm high}(e)\|\lesssim_{\tau} \|V\|_{\ell^{\frac{d+1}{2}}L^{\frac{d}{2}}}
\end{align}

\begin{proof}
  By a trivial $L^2$-bound we have
  \begin{align}
    \label{trivial L2 bound BShigh}
    \|BS^{\rm high}(e)\| \lesssim_{\tau} \||V|^{1/2}\langle\nabla\rangle^{-1}\|^2\,.
  \end{align}
  The $TT^*$ version of Lemma \ref{lemma Ionescu--Schlag}
  for $s=2$,
  \begin{align*}
    \|\langle\nabla\rangle^{-\frac{1}{d+1}}V\langle\nabla\rangle^{-\frac{1}{d+1}}\phi\|_{L^{\kappa}}
    \lesssim \|V\|_{\ell^{\frac{d+1}{2}}L^{\frac{d}{2}}}\|\phi\|_{L^{\kappa'}},
  \end{align*}
  together with Sobolev embedding $H^{\frac{d}{d+1}}(\R^d)\subset L^{\kappa'}(\R^d)$ yields
  \begin{align*}
    \|\langle\nabla\rangle^{-1}V\langle\nabla\rangle^{-1}\phi\|_{L^{2}}
    \lesssim \|V\|_{\ell^{\frac{d+1}{2}}L^{\frac{d}{2}}}\|\phi\|_{L^2}.
  \end{align*} 
  Combining the last inequality with \eqref{trivial L2 bound BShigh}
  yields the claim.
\end{proof}

\subsection{Bound for $BS^{\rm low}(e)$}

The Fermi surface of $T$ at energy $t\in (0,\tau]$ consists of
two connected components $S_t^{\pm}=(1\pm t)^{1/2}S$. The spectral
measure $E_{T}$ of $T$ is given by
\begin{align}
  \label{spectral measure 2}
  \rd E_{T}(t) = \sum_{\pm}\mathcal{F}_{S_t^{\pm}}^*\mathcal{F}_{S_t^{\pm}}\frac{\rd t}{2\sqrt{1\pm t}}.
\end{align}
in the sense of Schwartz kernels, see, e.g., \cite[Chapter XIV]{MR705278}.
By the spectral theorem, \eqref{spectral measure 2} implies that 
\begin{align}
  \label{BSlow spectral measure rep.}
  BS^{\rm low}(e)
  = \sum_{\pm}\int_0^{\tau}\frac{\sqrt{|V|}\mathcal{F}_{S_t^{\pm}}^*\mathcal{F}_{S_t^{\pm}}\sqrt{V}}{t+e}\,\frac{\rd t}{2\sqrt{1\pm t}}.
\end{align}
Together with the proof of Lemma \ref{lemma compactenss of V_S} this yields
\begin{align}
  \label{BSlow bound}
  \|BS^{\rm low}(e)\|_{\mathfrak{S}^{d+1}}
  \lesssim_\tau \ln(1/e) \|V\|_{\ell^{\frac{d+1}{2}}L^{\frac{d}{2}}}.
\end{align}

\subsection{Proof of the key bound \eqref{Key estimate}}\label{pfkeybound}

From \eqref{BSlow spectral measure rep.} and the definition of
$BS^{\rm low}_{\rm sing}(e)$ (see \eqref{BSlowsing def.}) we infer
that
\begin{align}
  \label{BSlowreg}
  BS^{\rm low}_{\rm reg}(e)
  = \sum_{\pm}\int_0^{\tau}\frac{\sqrt{|V|}(\mathcal{F}_{S_t^{\pm}}^*\mathcal{F}_{S_t^{\pm}}-\sqrt{1\pm t}\,\mathcal{F}^*_{S}\mathcal{F}_{S})\sqrt{V}}{t+e}\,\frac{\rd t}{2\sqrt{1\pm t}}.
\end{align}
If $V$ were a strictly positive Schwartz function, then by the
Sobolev trace theorem, the map
$t\mapsto\sqrt{V}\mathcal{F}_{S_t^{\pm}}^*\mathcal{F}_{S_t^{\pm}}\sqrt{V}$
would be Lipschitz continuous in operator norm, see, e.g.,\
\cite[Chapter 1, Proposition 6.1]{MR2598115},
\cite[Theorem IX.40]{MR0493420}. Hence, we would obtain a
stronger bound than \eqref{Key estimate} in this case. Using
\eqref{spectral measure 2} and observing that
\begin{align*}
  \mathcal{F}^*_{\mu S}\mathcal{F}_{\mu S}(x,y)
  = \mu^{d-1}\int_S\e^{2\pi \I\mu(x-y)\cdot\xi}\rd\omega(\xi)
\end{align*}   
for $\mu>0$, it is not hard to see that Lipschitz continuity even
holds in the Hilbert--Schmidt norm. Since
$\mathfrak{S}^2\subseteq\mathfrak{S}^{d+1}$ we conclude that, if
$V$ were Schwartz, we would get
\begin{align}
  \label{Key estimate Schatten norm}
  \lambda\|BS^{\rm low}_{\rm reg}(e)\|_{\mathfrak{S}^{d+1}}=o(1).
\end{align}
We now prove that \eqref{Key estimate Schatten norm} (and hence
also \eqref{Key estimate}) holds for the potentials considered
in Theorem \ref{thm. asymptotics}.

\begin{lemma}
  \label{proofkeybound}
  If $V\in \ell^{\frac{d+1}{2}}L^{\frac{d}{2}}$, then
  \eqref{Key estimate Schatten norm} holds as $\lambda\to 0$
  and $\lambda\ln(1/e)$ remains bounded.
\end{lemma}

\begin{proof}
  Without loss of generality we may again assume $V\geq 0$. Let
  $V_{n}^{1/2}$ be strictly positive Schwartz functions converging
  to $V^{1/2}$ in $\ell^{d+1}L^{d}$. We use that the bound
  \eqref{enhanced Tomas--Stein TT*} is locally uniform in $t$ and
  can be upgraded to a Schatten bound as in Lemma
  \ref{lemma compactenss of V_S}. That is, for fixed $\tau$, we
  have the bound
  \begin{align}
    \label{eq:tslocallyuniform}
    \sup_{t\in [0,\tau]}\|\sqrt{V}\mathcal{F}_{S_t}^*\mathcal{F}_{S_t}\sqrt{V}\|_{\mathfrak{S}^{d+1}}
    \lesssim_\tau \|V\|_{\ell^{\frac{d+1}{2}}L^{\frac{d}{2}}}.
  \end{align}
  Since we have already proved \eqref{Key estimate Schatten norm}
  for such $V_n$, we may thus estimate
  \begin{align*}
    \lambda\|BS^{\rm low}_{\rm reg}(e)\|_{\mathfrak{S}^{d+1}}
    \lesssim_\tau \lambda\ln(1/e) \|\sqrt{V}-\sqrt{V_n}\|_{\ell^{d+1}L^{d}}\|\sqrt{V}\|_{\ell^{d+1}L^{d}} + o(1)\,.
  \end{align*}
  Since $\lambda\ln(1/e)$ is bounded, \eqref{Key estimate Schatten norm}
  follows upon letting $n\to\infty$.
\end{proof}


\section{Further results}

The purpose of this subsection is fourfold.
First we outline how our main result, Theorem \ref{thm. asymptotics},
can be generalized to treat operators whose kinetic energy vanishes
on other smooth, curved surfaces.
Second, we provide an alternative proof (to that of
\cite{MR2365659,MR2643024}), based on Riesz projections, that weakly
coupled bound states of $H_{\lambda}=|\Delta+1|-\lambda V$ actually
exist, provided $\mathcal{V}_S$ has at least one positive eigenvalue.
This follows from standard perturbation theory
\cite[Sections IV.3.4-5]{Ka}, but the argument is robust enough to
handle complex-valued potentials. In fact, we do not know how the
arguments in \cite{MR2365659,MR2643024} could be adapted to treat
such potentials as the Birman--Schwinger operator cannot be
transformed to a self-adjoint operator anymore.
Third, we give two examples of (real-valued) potential classes for
which the operator $\mathcal{V}_S$ does have at least one positive
eigenvalue. In both examples the potentials are neither assumed to
be integrable, nor positive.
Fourth, we derive the second order in the asymptotic expansion of
$e_j(\lambda)$ in Theorem \ref{secondorder} for
$V\in L^{\frac{d+1}{2}-\epsilon}$ and $\epsilon\in(0,1/2]$.

\subsection{Generalization to other kinetic energies}
\label{generalkinen}

As the Tomas--Stein theorem holds for arbitrary compact, smooth,
curved surfaces (cf. \cite[Theorem 3]{Stein1986} and
\cite[Theorem 2]{MR3730931}) it is not surprising that Theorem
\ref{thm. asymptotics} continues to hold for more general symbols
$T(\xi)$. In what follows, we assume that $T(\xi)$ satisfies the
geometric and analytic assumptions stated in \cite{MR2643024}
-- that we recall in a moment -- and a certain curvature assumption.
First, we assume that $T(\xi)$ attains its minimum, which we set to
zero for convenience, on a manifold
\begin{align}
  S = \{\xi\in\R^d:T(\xi)=0\}
\end{align}
of codimension one. Next, we assume that $S$ consists of finitely
many connected and compact components and that there exists a
$\delta>0$ and a compact neighborhood $\Omega\subseteq\R^d$ of $S$
containing $S$ with the property that the distance of any point in
$S$ to the complement of $\Omega$ is at least $\delta$.

We now make some analytic assumptions on the symbol $T(\xi)$.
We assume that
\begin{enumerate}
\item there exists a measurable, locally bounded function
  $P\in C^\infty(\Omega)$
  such that $T(\xi)=|P(\xi)|$,
\item $|\nabla P(\xi)|>0$ for all $\xi\in\Omega$, and
\item there exist constants $C_1,C_2>0$ and $s\in[2d/(d+1),d)$ such
  that $T(\xi)\geq C_1|\xi|^s+C_2$ for $\xi\in\R^d\setminus\Omega$.
\end{enumerate}
Since $S$ is the zero set of the function $P\in C^\infty(\Omega)$
and $\nabla P\neq0$, it is a compact $C^\infty$ submanifold of
codimension one. Finally, we also assume that $S$ has everywhere
non-zero Gaussian curvature\footnote{The precise definition of
  Gaussian curvature can be found, e.g., in
  \cite[p. 321-322]{Stein1986}.}. Note that this assumption was
not needed in \cite{MR2643024}.

Next, we redefine the singular part of the Birman--Schwinger
operator \eqref{eq:lswbs}, namely
\begin{align}
  \label{eq:lswbsgen}
  (\mathcal{V}_Su)(\xi)
  = \int_S \widehat{V}(\xi-\eta) u(\eta)\,\rd\sigma_S(\eta)\,,
  \quad u\in L^2(S,\rd\sigma_S)\,.
\end{align}
Here, $\rd\sigma_S(\xi):=|\nabla P(\xi)|^{-1}\rd\omega(\xi)$ where
$\rd\omega$ denotes the euclidean (Lebesgue) surface measure on $S$.
In particular, the elementary volume $\rd\xi$ in $\R^d$ satisfies
$\rd\xi=\rd r\,\rd\sigma_S(\xi)$ where $\rd r$ is the Lebesgue measure
on $\R$. In what follows, we abbreviate the notation and write $L^2(S)$
instead of $L^2(S,\rd\sigma_S)$.

The new definition \eqref{eq:lswbsgen} of $\mathcal{V}_S$ now does
not differ anymore from that of \cite{MR2365659,MR2643024} by a factor
of $2$. Similarly as before, \eqref{eq:lswbsgen} can be written as
\begin{align}
  \label{eq:lswbs2gen}
  \cv_S = \mathcal{F}_{S}V\mathcal{F}_{S}^*\,,
\end{align}
where $\mathcal{F}_{S}:\cs(\R^d)\to L^2(S)$, $\phi\mapsto\widehat{\phi}|_S$
is the Fourier restriction operator and its adjoint, the Fourier
extension operator $\mathcal{F}_S^*:L^2(S)\to \cs'(\R^d)$, is now given by
\begin{align}
  (\mathcal{F}_S^* u)(x)
  = \int_S u(\xi) \me{2\pi ix\cdot\xi} \,\rd\sigma_S(\xi)\,.
\end{align}
Recall that the Tomas--Stein theorem asserts that $\mathcal{F}_S$ is
a $L^p(\R^d)\to L^2(S)$ bounded operator for all $p\in[1,\kappa]$. In
particular, the extension to trace ideals \cite[Theorem 2]{MR3730931}
continues to hold, i.e.,
$\|\cv_S\|_{\mathfrak{S}^{d+1}}\lesssim \|V\|_{L^{(d+1)/2}}$. By Sobolev
embedding and $s<d$, the operator $T(-i\nabla)-\lambda V$ can be
meaningfully defined if $V\in L^{d/s}(\R^d)$. By the assumption
$s\geq 2d/(d+1)$, we have $(d+1)/2\geq d/s$.

\medskip
We will now outline the necessary changes in the proof of Theorem
\ref{thm. asymptotics} for $T$ as above and
$V\in \ell^{\frac{d+1}{2}}L^{\frac{d}{s}}$. First, the corresponding analogs
of Proposition \ref{proposition def. V_Sw} and Lemma
\ref{lemma compactenss of V_S} follow immediately from the Tomas--Stein
theorem that we just discussed, and the analog of \eqref{eq:pfsmoothingbound}
(using $(d+1)/2\geq d/s$).

Next, the splitting of $BS(e)$ is the same as in Section
\ref{proofmainresult}. There we have the analogous bound
\eqref{BShigh bound}, i.e.,
$\|BS^{\mathrm{high}}(e)\|\lesssim_\tau\|V\|_{\ell^{(d+1)/2}L^{d/s}}$ by
the same arguments of that proof (cf. Lemma \ref{lemma Ionescu--Schlag}).
Next the Fermi surface of $T$ at energy $t\in(0,\tau]$ again consists of
two connected components $S_t^\pm$. Using the above definition of
$\mathcal{F}_{S_t^\pm}$, we observe that the spectral measure $E_T$ of
$T$ is now given by
\begin{align*}
  \rd E_T(t) = \sum_\pm \mathcal{F}_{S_t^\pm}^*\mathcal{F}_{S_t^\pm}\, \rd t\,.
\end{align*}
Thus, by the spectral theorem,
\begin{align*}
  BS^{\mathrm{low}}(e)
  = \sum_\pm\int_0^\tau \frac{\sqrt{|V|}\mathcal{F}^*_{S_t^\pm}\mathcal{F}_{S_t^\pm}\sqrt{V}}{t+e}\,\rd t
\end{align*}
and by the proof of the analog of Lemma \ref{lemma compactenss of V_S}
(i.e., the Tomas--Stein theorem and the analog of
\eqref{eq:pfsmoothingbound}), we again obtain
$\|BS^{\mathrm{low}}(e)\|_{\mathfrak{S}^{d+1}}\lesssim_\tau \ln(1/e)\|V\|_{\ell^{(d+1)/2}L^{d/s}}$.
Thus, we are left to prove the analog of the key bound
\eqref{Key estimate}. But this just follows from the proof
of Lemma \ref{proofkeybound} and the fact that the Tomas--Stein
estimate \eqref{eq:tslocallyuniform} is valid locally uniform in
$t$ for surfaces $S_t$ that we discuss here. In turn, by
\cite[Theorem 1.1]{MR2831841}, this is a consequence of the
following assertion.

\begin{proposition}
  \label{tsuniform}
  Assume $T(\xi)$ satisfies the assumptions stated at the
  beginning of this section. Then for fixed $\tau>0$, one has
  $\sup_{t\in[0,\tau]}|(\rd\sigma_{S_t^\pm})^\vee(x)|\lesssim_\tau(1+|x|)^{-\frac{d-1}{2}}$.
\end{proposition}

\begin{proof}
  For $t=0$ this estimate is well known, see, e.g.,
  \cite[Theorem 1]{Stein1986}.
  Now let $t\in(0,\tau]$. First note that $S_t=S_t^+\cup S_t^-$ where
  $S_t^{+},S_t^-$ lie outside, respectively inside $S$. In the following
  we treat $S_t^+$ and abuse notation by writing $S_t\equiv S_t^+$. The
  arguments for $S_t^-$ are completely analogous.
  We will now express $\rd\sigma_{S_t}$ in terms of $\rd\sigma_S$.
  To that end we follow \cite[Chapter 2, Section 1]{MR2598115}.
  Let $\psi(t):S\to S_t$ be the diffeomorphism\footnote{Its construction
    is carried out in \cite[p. 112-113]{MR2598115} and requires actually
    only $P\in C^2$. However, we need the smoothness of $P$ to obtain the
    claimed decay of $(\rd\sigma_{S_t})^\vee$ by means of a stationary
    phase argument.} defined by the formula
  \begin{align*}
    \psi(t)\zeta = \xi(t)\,, \quad \zeta\in S
  \end{align*}
  where $\xi(t)$ solves the differential equation
  \begin{align*}
    \begin{cases}
      \frac{\rd\xi(t)}{\dt} = j(\xi(t))\\
      \xi(0)=\zeta\in S
    \end{cases}
  \end{align*}
  with
  \begin{align*}
    j(\xi) := \frac{\nabla P(\xi)}{|\nabla P(\xi)|^2}
    \in C^\infty(P^{-1}[0,t])\,,
  \end{align*}
  i.e., $j(\xi(t))$ is the vector field generating the flow $\xi(t)$
  along the normals of $S_t$. Next,
  \begin{align*}
    \tau(t,\xi) = \frac{\rd\sigma_{S_t}(\psi(t)\xi)}{\rd\sigma_S(\xi)}\,,
    \quad \xi\in S
  \end{align*}
  is the Radon--Nikod\'ym derivative of the preimage of the measure
  $\rd\sigma_{S_t}$ under the mapping $\psi(t)$ with respect to the
  measure $\rd\sigma_S$. By \cite[Chapter 2, Lemma 1.9]{MR2598115} it
  is given by
  $$
  \tau(t,\xi) = \exp\left(\int_0^t (\Div j)(\psi(\mu)\xi)\,\rd\mu\right)\,,
  \quad \xi\in S\,.
  $$
  Thus, we have
  \begin{align}
    \label{eq:diffftmeas}
    \begin{split}
      (\rd\sigma_{S_t})^\vee(x)
      & = \int_S \rd\sigma_S(\xi)\,\me{2\pi ix\cdot\psi(t)\xi}\exp\left(\int_0^t \Div j(\psi(\mu)\xi)\,\rd\mu\right)\\
      & \equiv \int_S \rd\sigma_S(\xi)\,\me{2\pi ix\cdot\xi} F_{t,x}(\xi)
    \end{split}
  \end{align}
  with
  $$
  F_{t,x}(\xi) := \me{2\pi ix\cdot(\psi(t)\xi-\xi)}\exp\left(\int_0^t \Div j(\psi(\mu)\xi)\,\rd\mu\right)
  $$
  which depends smoothly on $\xi$.
  Thus, we are left to show that the absolute value of the right side
  of \eqref{eq:diffftmeas} is bounded by $C_\tau(1+|x|)^{-(d-1)/2}$ for
  all $t\in(0,\tau]$.
  Decomposing $F_{t,x}(\xi)$ on $S$ smoothly into (sufficiently small)
  compactly supported functions, say
  $\{F_{t,x}(\xi)\chi_\kappa(\xi)\}_{\kappa=1}^K$ for a finite, smooth
  partition of unity $\{\chi_\kappa\}_{\kappa=1}^K$ subordinate to $S$,
  shows that there is for every $x\in\R^d\setminus\{0\}$, at most one
  point $\overline\xi(x)\in S$ with a normal pointing in the direction
  of $x$. Then, by the stationary phase method, Hlawka \cite{Hlawka1950}
  and Herz \cite{Herz1962} (see also Stein \cite[p. 360]{Stein1993})
  already showed that the leading order in the asymptotic expansion
  (as $|x|\to\infty$) of \eqref{eq:diffftmeas} with the cut off
  amplitude $F_{t,x}\chi_\kappa$ is given by
  $$
  |x|^{-(d-1)/2} F_{t,x}(\overline{\xi}(x))\chi_\kappa(\overline{\xi}(x)) |K(\overline{\xi}(x))|^{-1/2}\me{i\pi n/4+2\pi i x\cdot\overline{\xi}(x)}\,.
  $$
  Here, $|K(\xi)|$ is the absolute value of the Gaussian curvature of
  $S$ at $\xi\in S$ which is, by assumption, strictly bigger than zero,
  and $n$ is the excess of the number of positive curvatures over the
  number of negative curvatures in the direction $x$.
  But since $|F_{t,x}(\xi)|\lesssim_\tau1$ for all $t\in(0,\tau]$,
  $x\in\R^d$, and $\xi\in S$, this concludes the proof.
\end{proof}

We summarize the findings of this subsection as follows.

\begin{theorem}
  \label{asymptoticsgen}
  Let $d\geq2$, $s\in[2d/(d+1),d)$, and assume $T(\xi)$ satisfies
  the assumptions stated at the beginning of this subsection.
  If $V\in\ell^{\frac{d+1}{2}}L^{\frac{d}{s}}$, then for every eigenvalue
  $a_S^j>0$ of $\cv_S$ in \eqref{eq:lswbs2gen}, counting multiplicity,
  and every $\lambda>0$, there is an eigenvalue $-e_j(\lambda)$ of
  $T(-i\nabla)-\lambda V$ with weak coupling limit
  \begin{align*}
    e_j(\lambda) = \exp\left(-\frac{1}{2\lambda a_S^j}(1+o(1))\right)
    \quad \text{as}\ \lambda\to0\,.
  \end{align*}
\end{theorem}

\subsection{Alternative proof and complex-valued potentials}
\label{Section alternative proof} 
We first consider the case where $V$ is real-valued and then indicate
how to modify the proof in the complex-valued case. For simplicity,
we even assume $V\geq 0$ so that the Birman--Schwinger operator is
automatically self-adjoint. The case where $V$ does not have a sign
could also be treated by the methods of \cite{MR2365659}, but here it
follows from the general case considered later. 

For $V\geq 0$ we have, by self-adjointness, 
\begin{align}\label{bound for inverse of BS}
\|(BS^{\rm low}_{\rm sing}(e)-z)^{-1}\|\leq 1/\min_j|z-z_j(e)|,
\end{align}
where $z_j(e)=\ln\left(1+\tau/e\right)a_S^j$ are the eigenvalues of
$BS^{\rm low}_{\rm sing}(e)$. Fixing an integer $i$ and a range for $e$
such that $\lambda\ln(1/e)$ is bounded by an absolute constant from
above and below, it follows that if $\gamma$ is a circle of radius
$c\ln(1/e)$ around the eigenvalue $z_i(e)$, with $c$ a sufficiently
small positive number, then there are no other eigenvalues in the
interior of $\gamma$, and 
\begin{align*}
  \max_{z\in\gamma}\|(BS^{\rm low}_{\rm sing}(e_i(\lambda))-z)^{-1}\|
  \leq 1/(c\ln(1/e)).
\end{align*}
Hence, by \eqref{BShigh bound} and \eqref{Key estimate}, if we set
$C(z)=(BS^{\rm low}_{\rm sing}(e)-z)^{-1}(BS(e)-BS^{\rm low}_{\rm sing}(e))$,
then
\begin{align}\label{<1!}
r^{-1} := \max_{z\in\gamma}\|C(z)\|\leq c^{-1} o(1),
\end{align}
and this is $<1$ for $\lambda$ small enough. It follows from a
Neumann series argument that $\gamma$ is contained in the resolvent
set of the family
$T(\kappa)=BS^{\rm low}_{\rm sing}(e)+\kappa(BS(e)-BS^{\rm low}_{\rm sing}(e))$
for $|\kappa|<r$ and that $(T(\kappa)-z)^{-1}$ is continuous in
$|\kappa|<r$, $z\in\gamma$. This implies that the Riesz projection
\begin{align*}
  P(\kappa)=-\frac{1}{2\pi\I}\oint_{\gamma}(T(\kappa)-z)^{-1}\rd z
\end{align*}
has constant rank for $|\kappa|<r$. In particular, $\rk P(0)=\rk P(1)$,
which means that $BS^{\rm low}_{\rm sing}(e)$ and $BS(e)$ have the same
number of eigenvalues in the interior of $\gamma$. Hence $BS(e)$ has
exactly one (real) eigenvalue $w_i(e)$ at a distance $\leq c\ln(1/e)$
from $z_i(e)$. Since $c$ can be chosen arbitrarily small, it follows
that $w_i(e)=z_i(e)(1+o(1))$. By the Birman--Schwinger principle this
implies \eqref{eq. weak coupling limit}.

We now drop the assumption that $V$ is real-valued. By inspection of
the proof, it is evident that Lemma \ref{lemma bound BS(e)} and
\eqref{Key estimate Schatten norm} continue to hold for complex-valued
$V$ and $e$ if $\ln(1/e)$ is replaced by its absolute value. We assume
here that $e\in \C\setminus(-\infty,0]$ and take the branch of the
logarithm that agrees with the real logarithm on the positive real line.
We also replace our standing assumption by requiring that
$|e|,\lambda>0$ are sufficiently small and $\lambda|\ln(e)|$ remains
uniformly bounded from above and below. The additional difficulty in
the present case is that the bound for the inverse
\eqref{bound for inverse of BS} fails in general. We use the following
replacement, which is a consequence of \cite[Theorem 4.1]{MR2047381},
\begin{align*}
  \|(BS^{\rm low}_{\rm sing}(e)-z)^{-1}\|
  \leq \frac{1}{d(e;z)}\exp\left(a\,\frac{\|BS^{\rm low}_{\rm sing}(e)\|_{\mathfrak{S}^{d+1}}^{d+1}}{d(e;z)^{d+1}}+b\right),
\end{align*}
where $d(e;z)=\dist(z,\mathrm{spec}(BS^{\rm low}_{\rm sing}(e)))$ and
$a,b>0$. Note that
$\|BS^{\rm low}_{\rm sing}(e)\|_{\mathfrak{S}^{d+1}}\lesssim |\ln(1/e)|\|V\|_{\ell^{\frac{d+1}{2}}L^{\frac{d}{2}}}$
by Lemma \ref{lemma compactenss of V_S}. Thus, for a similar circle
$\gamma$ of radius $c|\ln(1/e)|$ around $z_i(e)$, we find that
\eqref{<1!} holds with an additional factor of $\exp(a/c^{d+1}+b)$ on
the right, and hence we conclude $\rk P(0)=\rk P(1)$ as before.

\subsection{Existence of positive eigenvalues of $\mathcal{V}_S$}
It is well known that operators of the form \eqref{eq:defhlambda}
have at least one negative eigenvalue if either $V\in L^1(\R^d)$
and $\int V>0$ or if $V\geq 0$ and not almost everywhere vanishing
\cite{MR1970614,MR2365659,MR2643024,hoang2016quantitative}. In the
latter case there are even infinitely many negative eigenvalues
\cite[Corollary 2.2]{MR2643024}. By Theorem \ref{thm. asymptotics},
$H_\lambda$ has at least as many negative eigenvalues as
$-\mathcal{V}_S$. We will therefore restrict our attention to this
operator. By a slight modification of the following two examples
(where the trial state is an approximation of the identity in
Fourier space to a thickened sphere), this result may also be
obtained without reference to Theorem \ref{thm. asymptotics}.

Since $\mathcal{F}_S^*\phi=(\phi\rd\omega)^{\vee}$ it follows from
\eqref{eq:lswbs2} that
\begin{align}
  \label{phiV_Sphi}
  \langle \phi,\mathcal{V}_S\phi\rangle
  = \int_{\R^d}V(x)|(\phi\rd\omega)^{\vee}(x)|^2\rd x,\quad \phi\in L^2(S).
\end{align} 
If $\phi$ is a radial function, then so is $(\phi\rd\omega)^{\vee}$.
In particular, for $\phi\equiv 1$ we get 
\begin{align*}
  \langle \phi,\mathcal{V}_S\phi\rangle
  = \int_0^\infty \rd r\, r^{d-1}|(\rd\omega)^{\vee}(r)|^2\left(\int_S V(r\omega)\rd\omega\right).
\end{align*}
Standard stationary phase computations show that
$(\rd\omega)^{\vee}(r)=\mathcal{O}((1+r)^{-(d-1)/2})$ and that it
oscillates on the unit scale; in fact, it is proportional to the
Bessel function $J_{\frac{d-2}{2}}$, see, e.g.,
\cite[Appendix B.5]{MR3243734}. The integral is convergent if the
spherical average of $V$ is in $L^1(\R_+,\min\{r^{d-1},1\}\rd r)$.
This condition is satisfied, e.g., if $V$ is short range,
$|V(x)|\lesssim (1+|x|)^{-1-\epsilon}$ for some $\epsilon>0$. If the
integral is positive, then $\mathcal{V}_S$ has a positive eigenvalue. 

For the second example we take $\phi$ as a normalized bump function
adapted to a spherical cap of diameter $R^{-1/2}$ with $R>1$; this is
called a Knapp example in the context of Fourier restriction theory.
Then $(\phi\rd\omega)^{\vee}$ will be a Schwartz function concentrated
on a tube $T=T_{R}$ of length $R$ and radius $R^{1/2}$, centered at
the origin. More precisely, let 
\begin{align}
  \label{normalized bump adapted to cap}
  \phi(\xi) = R^{\frac{d-1}{4}}\widehat{\chi}(R(\xi_1-1),R^{1/2}\xi')
\end{align}
where $\xi_1=\sqrt{1-|\xi'|^2}$ and $\widehat{\chi}$ is a bump
function. We write $\xi=(\xi_1,\xi')\in \R\times\R^{d-1}$ and similarly
for $x$ here. We may choose $\chi\geq 0$ and such that
$\chi\geq \mathbf{1}_{B(0,1)}$. Indeed, if $g$ is an even bump function,
then we can take $\widehat{\chi}(\xi)=A^dB(g*g)(A\xi)$ for some
$A>1,B>0$. Then the $L^2(S)$-norm of $\phi$ is bounded from above and
below uniformly in $R$ and
\begin{align*}
  (\phi\rd\omega)^{\vee}(x)=R^{-\frac{d-1}{4}}e^{2\pi i x_1}\chi_{T}(x),
\end{align*}
where $\chi_T$ is a Schwartz function concentrated on
\begin{align}
  \label{tube}
  T=\{x\in\R^d:\,|x_1|\leq R,|x'|\leq R^{1/2}\},
\end{align}
i.e., a tube pointing in the $x_1$ direction. We can also take linear
combinations of the wave packets \eqref{normalized bump adapted to cap}
to obtain real-valued trial functions. Indeed, choosing $\chi$ symmetric
and setting $\psi(\xi)=[\phi(\xi_1,\xi')+\phi(-\xi_1,\xi')]/2$, we get
\begin{align*}
  (\psi\rd\omega)^{\vee}(x) = R^{-\frac{d-1}{4}}\cos(2\pi x_1)\chi_{T}(x),
\end{align*}
with a slightly different $\chi_T$. Without loss of generality we may
assume that $\chi_{T}(x)\geq 1$ for $x\in T$. By \eqref{phiV_Sphi}, if
$V\in L^1_{\rm loc}(\R^d)$ and of tempered growth, then
\begin{align*}
  \langle \psi,\mathcal{V}_S\psi\rangle
  = R^{-\frac{d-1}{2}}\int_{\R^d} V(x)\cos^2(2\pi x_1)|\chi_{T}(x)|^2\rd x.
\end{align*}
In particular, this holds for $V\in \ell^{\frac{d+1}{2}}L^{\frac{d}{2}}$,
which we assume from now on. By H\"older and the rapid decay of $\chi_T$
away from $T$, we have that, for any $M,N>1$,
\begin{align*}
  |\int_{\R^d\setminus MT}V(x)|\chi_{T}(x)|^2\rd x|
  \leq \|\mathbf{1}_{\R^d\setminus MT}\chi_T^2\|_{\ell^{\frac{d+1}{d-1}}L^{\frac{d}{d-2}}}\|V\|_{\ell^{\frac{d+1}{2}}L^{\frac{d}{2}}}
  \lesssim_{N}M^{-N}R^{\frac{d-1}{2}}\|V\|_{\ell^{\frac{d+1}{2}}L^{\frac{d}{2}}}.
\end{align*}
It follows that for $V\in \ell^{\frac{d+1}{2}}L^{\frac{d}{2}}$,
\begin{align}
  \label{second example}
  \langle \psi,\mathcal{V}_S\psi\rangle
  \geq R^{-\frac{d-1}{2}}\int_{MT} V(x)\cos^2(2\pi x_1)|\chi_T(x)|^2\rd x-C_{N}M^{-N}\|V\|_{\ell^{\frac{d+1}{2}}L^{\frac{d}{2}}}.
\end{align}
If the first term on the right is positive and bounded from below
by, say, a fixed power of $M^{-1}$, then this expression is positive
for large $R$. As a concrete example, consider the potential
\begin{align*}
  V(x) = \frac{\cos(4\pi x_1)}{(1+|x_1|+|x'|^2)^{1+\epsilon}},
\end{align*}
with $\epsilon>0$ (see also \cite{MR2024415,MR3713021,arXiv:1709.06989}
for related examples). A straightforward calculation shows that
$V\in \ell^{\frac{d+1}{2}}L^{\frac{d}{2}}$. Since the average of
$\cos^2(2\pi x_1)\cos(4\pi x_1)$ over a full period of
$\cos(4\pi x_1)$ is always $\gtrsim 1$ and $|\chi_T|^2$ is
approximately constant on the unit scale, with $\geq 1$ on $T$,
a computation shows that the first term on the right side of
\eqref{second example} is bounded from below by $MR^{-\epsilon}$.
Taking $M=R^{\epsilon}$ yields positivity of the whole expression
for sufficiently large $R$. Therefore, $-\mathcal{V}_S$, and hence
$H_{\lambda}$, has a negative eigenvalue. This example has a
straightforward generalization to more than one eigenvalue. Let
$(\kappa_j)_{j=1}^K$ be mutually disjoint spherical caps of diameter
$R^{-1/2}$ and let $\phi_j$ be normalized bump functions adapted to
$\kappa_j$, similar to \eqref{normalized bump adapted to cap}. Note
that $K\lesssim R^{-\frac{d-1}{2}}$ since the caps are disjoint. If
the condition following \eqref{second example} is satisfied for all
tubes $T_j$ corresponding to the caps $\kappa_j$ (these are dual to
the caps and centered at the origin), then the expression
\eqref{second example} is positive (for large $R$) for every $\phi_j$.
Since the $\phi_j$ are orthogonal (by Plancherel), it follows that
$\mathcal{V}_S$ has at least $K$ positive eigenvalues.

\subsection{Higher orders in the eigenvalue asymptotics}
\label{ss:higherorders}

Hainzl and Seiringer carried out the higher order asymptotic expansion
of the eigenvalues $e_j(\lambda)$ in \cite[Formula (16)]{PhysRevB.77.184517}
and \cite[Theorem 2.7]{MR2643024} under the assumption that $V$ has an
$L^1$ tail.
Similarly as in Theorems \ref{thm. asymptotics} and \ref{asymptoticsgen},
the purpose of this section is to show that their findings in fact hold
for potentials decaying substantially slower.
For the sake of simplicity and concreteness, we again only consider
$T=|\Delta+1|$ here.

Let $BS_\reg(e)=BS(e)-BS_\sing^\low(e)$ and recall that if
$1+\lambda BS_\reg(e)$ is invertible, then the Birman--Schwinger principle
\eqref{eq:bsprinc} asserts that $H_\lambda$ has a negative eigenvalue $-e$
if and only if the operator
\begin{align}
  \label{eq:bsprinciplerewritten2}
  \frac{\lambda}{1+\lambda BS_\reg(e)} BS_\sing^\low(e)
\end{align}
has an eigenvalue $-1$. The following is a simple but useful observation
which follows from a Neumann series argument and the fact that
\eqref{eq:bsprinciplerewritten2} is isospectral to
\begin{align*}
  \ln(1+\tau/e)\F_S V^{1/2} \frac{\lambda}{1+\lambda BS_\reg(e)}|V|^{1/2}\F_S^*\,.
\end{align*}

\begin{lemma}
  \label{neumann}
  Let $e,\lambda>0$ and suppose $V$ is real-valued and such that
  \begin{align}
    \label{eq:assumbssmall}
    \lambda\|BS_\reg(e)\|<1\,.
  \end{align}
  Then $H_\lambda$ has an eigenvalue $-e$ if and only if
  \begin{align}
    \label{eq:neumann}
    \lambda\ln(1+\tau/e)\F_S V^{1/2} \left(\sum_{n\geq0}(-1)^n(\lambda BS_\reg(e))^n\right) |V|^{1/2}\F_S^*
  \end{align}
  has an eigenvalue $-1$.
\end{lemma}

Recall that assumption \eqref{eq:assumbssmall} is satisfied for
$V\in\ell^{\frac{d+1}{2}}L^{\frac{d}{2}}$ (cf. \eqref{BShigh bound} and Lemma
\ref{proofkeybound}), i.e., in particular for $V\in L^{\frac{d+1}{2}-\epsilon}$
with $\epsilon\in(0,1/2]$. In fact, combining Lemma \ref{proofkeybound} for
$BS_\reg^\low(e)$ and the Seiler--Simon inequality (cf.
\cite[Theorem 4.1]{MR2154153}) for $BS^\high(e)$ shows\footnote{Using Cwikel's
  inequality \cite[Theorem 4.2]{MR2154153}, one obtains
  $\|BS_\reg(e)\|_{\mathfrak{S}^{\frac{(d-1)((d+1)/2-\epsilon)}{(d-1)/2+\epsilon},\infty}}=o_V(\ln(1/e))$
  for $\epsilon=1/2$.}

\begin{align}
  \label{eq:bsreglpschatten}
  \begin{split}
    \|BS_\reg(e)\|_{\mathfrak{S}^{\frac{(d-1)((d+1)/2-\epsilon)}{(d-1)/2+\epsilon}}}
    & \leq \|BS_\reg^\low(e)\|_{\mathfrak{S}^{\frac{(d-1)((d+1)/2-\epsilon)}{(d-1)/2+\epsilon}}}
    + \|BS^\high(e)\|_{\mathfrak{S}^{\frac{(d-1)((d+1)/2-\epsilon)}{(d-1)/2+\epsilon}}}\\
    & = o_V(\ln(1/e))\,,
    \quad V\in L^{\frac{d+1}{2}-\epsilon}\,, \quad \epsilon\in(0,1/2)\,.
  \end{split}
\end{align}

We will now use \eqref{eq:neumann} to compute the eigenvalue asymptotics
of $e_j(\lambda)$ to second order. To that end, we define
\begin{align}
  \begin{split}
    \cw_S(e) :&= \F_S V^{1/2}BS_\reg(e) |V|^{1/2}\F_S^*
  \end{split}
\end{align}
which is, modulo the $-\lambda^2\ln(1+\tau/e)$ prefactor, just the second
summand in \eqref{eq:neumann}. Note that due to the additional operators
$\F_S V^{1/2}$ on the left and $|V|^{1/2}\F_S^*$ on the right of $BS_\reg(e)$,
estimate \eqref{eq:tsfsexplicitgen}, and $\lambda\|BS_\reg(e)\|=o_V(1)$, we
infer
\begin{align}
  \label{eq:wsschatten}
  \|\cw_S(e)\|_{\mathfrak{S}^{\frac{(d-1)((d+1)/2-\epsilon)}{(d-1)/2+\epsilon}}}
  = o_V(\ln(1/e))\,,
  \quad V\in L^{\frac{d+1}{2}-\epsilon}\,, \quad \epsilon\in(0,1/2]\,.
\end{align}
We will momentarily show the existence of $\cw_S(0)$ and the limit
$\lim_{e\searrow0}\cw_S(e)=\cw_S(0)$ in operator norm for
$V\in L^{\frac{d+1}{2}-\epsilon}$. Let $b_S^j(\lambda)<0$ denote the
negative eigenvalues of
\begin{align}
  \label{eq:limitingopsecondorder}
  \cb_S(\lambda) := \cv_S - \lambda\cw_S(0) \quad \text{on}\ L^2(S)
\end{align}
and recall that
$\cv_S\in\mathfrak{S}^{\frac{(d-1)((d+1)/2-\epsilon)}{d-(d+1)/2+\epsilon}}$
if $V\in L^{\frac{d+1}{2}-\epsilon}$ by \eqref{eq:tsfsexplicitgen}.
This and \eqref{eq:wsschatten} show that $\cb_S(\lambda)$ is a compact
operator as well. Note that, by the definition of $BS_\reg(e)$, the
operator $\cb_S(\lambda)$ has at least one negative eigenvalue if
$\cv_S$ has a zero-eigenvalue. The asymptotic expansion of $e_j(\lambda)$
to second order then reads as follows.

\begin{theorem}
  \label{secondorder}
  Let $d\geq3$ and $V\in L^{\frac{d+1}{2}-\epsilon}$ for some
  $\epsilon\in(0,1/2]$. If $\lim_{\lambda\searrow0}b_S^j(\lambda)<0$
  then $H_\lambda$ has, for small $\lambda$, a corresponding negative
  eigenvalue $-e_j(\lambda)<0$ that satisfies
  \begin{align}
    \lim_{\lambda\to0}\left(\ln(1+1/e_j(\lambda))+\frac{1}{\lambda b_S^j(\lambda)}\right) = 0\,.
  \end{align}
\end{theorem}

The proof of Theorem \ref{secondorder} relies on the fact that
$|V|^{1/2}(\F_{S_t^\pm}^*\F_{S_t^\pm}-\sqrt{1\pm t}\F_S^*\F_S)V^{1/2}$
is H\"older continuous in $\mathcal{B}(L^2(\R^d),L^2(\R^d))$ for
$t\leq\tau\in(0,1)$. We already saw in Subsection \ref{pfkeybound}
that this is true for $V\in\cs(\R^d)$ (or more generally $V$ satisfying
$|V(x)|\lesssim(1+|x|)^{-1-\epsilon}$) because of H\"older continuity of
the Sobolev trace theorem. The following proposition, whose proof is
deferred to Appendix \ref{a:tsholdercont}, yields H\"older continuity
of the (non-endpoint) Tomas--Stein theorem.

\begin{proposition}
  \label{holdercontts}
  Let $0<\tau<1$, $1\leq p<\kappa$, $1/q=1/p-1/p'$, i.e., $1\leq q<(d+1)/2$,
  and $0<\alpha<\min\{(d+1)/2-q,q\}$. Then 
  \begin{align}
    \label{eq:holdercont}
    \sup_{t\in(0,\tau)}\|\F_{S_t^\pm}^*\F_{S_t^\pm}-\sqrt{1\pm t}\F_S^*\F_S\|_{L^{p}\to L^{p'}}
    \lesssim_{\alpha,q,\tau} t^{\alpha/q}\,.
  \end{align}
\end{proposition}

\begin{proof}[Proof of Theorem \ref{secondorder}]
  Recall that $V\in L^{\frac{d+1}{2}-\epsilon}$ satisfies the assumption of
  Lemma \ref{neumann}. Thus, $H_\lambda$ has an eigenvalue $-e_j(\lambda)<0$
  if and only if
  \begin{align}
    \begin{split}
      \lambda\ln(1+\tau/e_j(\lambda)) & [\cb_S(\lambda) + \lambda(\cw_S(0)-\cw_S(e_j(\lambda)))\\
      & \quad + \F_SV^{1/2}\left(\sum_{n\geq2}(-1)^n(\lambda BS_\reg(e_j(\lambda)))^n\right)|V|^{1/2}\F_S^*]
    \end{split}
  \end{align}
  has an eigenvalue $-1$. Thus, our claim is established, once we show 
  $\lim_{e\to0}\cw_S(e)=\cw_S(0)$ in operator norm topology. In turn,
  by the definition of $\cw_S(e)$, this follows once we show the existence of
  $\lim_{e\to0}BS_\reg(e)=BS_\reg(0)$ since $|V|^{1/2}\F_S^*$ and $\F_SV^{1/2}$
  are bounded by the Tomas--Stein theorem \eqref{eq:tsexplicit}.
  We decompose $BS_\reg(e)=BS^\high(e)+BS_\reg^\low(e)$ and observe that
  $BS^\high(e)\to BS^\high(0)$ (e.g., by Plancherel and dominated convergence).
  On the other hand, Proposition \ref{holdercontts} shows that the difference
  \begin{align*}
    & BS_\reg^\low(e) - BS_\reg^\low(0)\\
    & \quad = \sum_{\pm}\int_0^{\tau} \left(\sqrt{|V|}(\F_{S_t^{\pm}}^*\F_{S_t^{\pm}}-\sqrt{1\pm t}\,\F^*_{S}\F_{S})\sqrt{V}\right)\left(\frac{1}{t+e}-\frac{1}{t}\right)\,\frac{\rd t}{2\sqrt{1\pm t}}
  \end{align*}
  vanishes in operator norm as $e\to0$. This concludes the proof.
\end{proof}

\appendix
\section{H\"older continuity of the Tomas--Stein theorem}
\label{a:tsholdercont}

In this section we prove Proposition \ref{holdercontts} on the H\"older
continuity of the Tomas--Stein theorem for the sphere. The arguments can
be generalized to treat arbitrary smooth, curved, and compact hypersurfaces
by refining the analysis in the proof of Proposition \ref{tsuniform}.
However, for the sake of simplicity, we restrict ourselves to the unit
sphere $S=\bs^{d-1}$.

\begin{lemma}
  \label{holdercontabstract}
  Let $\tau\in(0,1)$, $0<\beta\leq(d-1)/2$, $p_\circ=2(1+\beta)/(2+\beta)\in(1,2)$,
  and $1\leq p<p_\circ$, and denote $1/q:=1/p-1/p'$ and
  $\widehat{\rd\omega^{\pm}}:=\widehat{\rd\omega_{S_t^{\pm}}}-\widehat{\rd\omega_S}$.
  If there is $\alpha\in(0,\min\{\beta+1-q,q\})$ such that
  \begin{align}
    \label{eq:ptboundmu}    
    |\widehat{\rd\omega^{\pm}}(x)|
    \leq c_{\tau}\, t^\alpha(1+|x|)^{\alpha-\beta}
  \end{align}
  holds for some $c_{\tau}>0$ and all $t\in(0,\tau)$, then 
  \begin{align}
    \label{eq:holdercontabstract}
    \sup_{t\in(0,\tau)}\|\F_{S_t^\pm}^*\F_{S_t^\pm}-\sqrt{1\pm t}\F_S^*\F_S\|_{L^{p}\to L^{p'}}
    \lesssim_{\alpha,q,\tau} t^{\alpha/q}\,.
  \end{align}
\end{lemma}

\begin{proof}
  In the following, we consider only $S_t^+$ and write $S_t\equiv S_t^+$
  and $\rd\omega\equiv\rd\omega^+$. As in Tomas' proof, we decompose
  \begin{align}
    \label{eq:holdercontabstractaux}
    \begin{split}
      (\F_{S_t}^*\F_{S_t}-\sqrt{1+t}\F_{S}^*\F_{S})f
      & = [(\F_{S_t}^*\F_{S_t}-\F_{S}^*\F_{S}) - (\sqrt{1+t}-1)\F_S^*\F_S]f\\
      & = \widehat{\rd\omega}*f - (\sqrt{1+t}-1)\F_S^*\F_Sf\\
      & = \sum_{k=0}^{\infty}(\widehat{\rd\omega}\psi_k)*f - (\sqrt{1+t}-1)\F_S^*\F_S f\\
      & =: \sum_{k=0}^{\infty}T_kf - (\sqrt{1+t}-1)\F_S^*\F_S f\,,
    \end{split}
  \end{align}
  where $(\psi_k)_k$ is a standard dyadic partition of unity such that
  $\psi_0$ is adapted to the unit ball and $\psi_k$ is adapted to the
  annulus $2^{k-1}\leq |x|\leq 2^{k+1}$.
  By the Tomas--Stein theorem, the operator norm of the second term on the
  right side of \eqref{eq:holdercontabstractaux} is bounded by a constant
  times
  \begin{align*}
    |\sqrt{1+t}-1|\|\F_S^*\F_S\|_{{L^{p}\to L^{p'}}} \lesssim t\,.
  \end{align*}
  We now focus on the first term on the right side of
  \eqref{eq:holdercontabstractaux}.
  By the triangle inequality, Plancherel, the rapid decay of $\hat\psi_k$,
  and the fact that $\rd\omega_{S_t}$ is a $(d-1)$-dimensional measure, we
  estimate
  \begin{align}
    \label{eq:boundtk0}
    \|T_k\|_{2\to 2}\lesssim 2^k\,,
  \end{align}
  whereas we use \eqref{eq:ptboundmu} to bound
  \begin{align}
    \label{eq:boundtk1}
    \|T_k\|_{1\to \infty}\lesssim t^{\alpha}\, 2^{-k\left(\beta-\alpha\right)}\,.
  \end{align}
  Interpolating between those two bounds yields
  \begin{align*}
    \|T_k\|_{p\to p'}\lesssim 2^{k(1-\theta)} \cdot t^{\alpha\theta}\, 2^{-k\left(\beta-\alpha\right)\theta},
  \end{align*}
  where $1-\theta=2/p'$, i.e., $\theta=1/q=1/p-1/p'$.
  Thus, 
  \begin{align*}
    \|T_k\|_{p\to p'}\lesssim t^{\alpha/q}\, 2^{k\left(1-\frac{\beta+1-\alpha}{q}\right)}\,.
  \end{align*}
  Since the exponent of $2^k$ is negative for $\alpha<\beta+1-q$,
  we obtain
  \begin{align*}
    \|\sum_{k\geq0}T_k\|_{p\to p'}
    \lesssim t^{\alpha/q}\sum_{k=0}^{\infty} 2^{k\left(1-\frac{\beta+1-\alpha}{q}\right)}
    \lesssim t^{\alpha/q}\,,
  \end{align*}
  which concludes the proof of \eqref{eq:holdercontabstract}.
\end{proof}

\begin{proof}[Proof of Proposition \ref{holdercontts}]
  For $T(\xi)=|\xi^2-1|$ we have 
  \begin{align*}
    S_0 = \mathbb{S}^{d-1} := S\,, \quad
    S_t = \sqrt{1-t}S \cup \sqrt{1+t}S\,, \quad 0<t\leq \tau<1\,.
  \end{align*}
  Setting $\rho=\sqrt{1\pm t}\in[\sqrt{1-\tau},\sqrt{1+\tau}]$,
  Lemma \ref{holdercontabstract} with $\beta=(d-1)/2$, i.e., $p_\circ=\kappa$
  shows that \eqref{eq:holdercont} would follow from
  \begin{align}
    \label{eq:holdercont2}
    |\widehat{\rd\omega_{\rho S}}(x)-\widehat{\rd\omega_{S}}(x)|
    \lesssim |\rho-1|^{\alpha}(1+|x|)^{\alpha-\frac{d-1}{2}}
  \end{align}
  for some $\alpha\in(0,\min\{(d-1)/2+1-q,q\})$.
  We have
  \begin{align}
    \label{eq:ptwiseboundtransported}
    \widehat{\rd\omega_{\rho S}}(x)
    = \int_{\rho S}\me{2\pi i x\cdot\xi}\rd\omega_{\rho S}(\xi)
    =\rho^{d-1}\widehat{\rd\omega_{S}}(\rho x)\,, \quad x\in\R^d\,.
  \end{align}
  The classic stationary phase argument (see, e.g., Stein
  \cite[Theorem 1]{Stein1986}) yields
  \begin{align*}
    |\widehat{\rd\omega_{S}}(x)|\lesssim (1+|x|)^{-\frac{d-1}{2}}\,,
  \end{align*}
  with the same bound for $|\nabla\widehat{\rd\omega_{S}}(x)|$.
  Combining this with \eqref{eq:ptwiseboundtransported} yields
  \eqref{eq:holdercont2} for $\alpha=0$ and $\alpha=1$ and hence
  for any $\alpha\in[0,1]$, thereby concluding the proof.
\end{proof}

\section{Further $L^2$ based restriction estimates }
\label{a:mt}

Throughout this appendix we take $T(\xi)=|\xi^2-1|$.
Our main results crucially relied on the fact that $\F_{S_t}^*\F_{S_t}$
belongs to $\mathcal{B}(L^p\to L^{p'})$ or
$\mathcal{B}(L^2((1+|x|)^{1+\epsilon}\,\dx)\to L^2((1+|x|)^{-1-\epsilon}\,\dx)$
uniformly for $t\in[0,\tau]$ and some fixed $\tau>0$.
Besides the Tomas--Stein estimate or the trace lemma, there exist
various other $L^2$-based restriction estimates.
In this appendix we present two such estimates and apply them
to obtain corresponding upper bounds on the Birman--Schwinger
operator, whenever the potential belongs to the suitable dual space.
It turns out that this space contains spherically symmetric
$V$ with almost $L^d$ decay, but see Propositions \ref{vegarestriction}
and \ref{mtrestriction} below for the precise assumptions. By the
uniformity of these restriction theorems on small compact sets around
$\bs^{d-1}$ (see \eqref{eq:vegarescaled} and \eqref{eq:barcelo}) and
following the arguments in Sections \ref{s:prelims} and
\ref{proofmainresult}, one obtains the analog of Theorem
\ref{thm. asymptotics} for potentials living in the spaces mentioned
below. This is the content of Theorem \ref{asymptoticsradial}.

\subsection{Estimates for potentials in mixed norm spaces}

\label{ss:vega}

Vega \cite{Vega1992} observed that the Tomas--Stein estimate
can be enhanced for $\bs^{d-1}$ if one replaces the $L^p$ spaces by
suitable mixed norm spaces. For $k>0$ recall the extension operator
\begin{align*}
  (\F_{k\bs^{d-1}}^*g)(x)
  = \int_{k\bs^{d-1}}g(\xi)\me{2\pi ix\cdot\xi}\,\rd\omega_{k\bs^{d-1}}(\xi)
  = k^{d-1}\int_{\bs^{d-1}}g(k\xi)\me{2\pi ikx\cdot\xi}\,\rd\omega(\xi)
\end{align*}
where $\rd\omega$ and $\rd\omega_{k\bs^{d-1}}$ denote the euclidean
surface measures on $\bs^{d-1}$ and $k\bs^{d-1}$, respectively.

\begin{theorem}[{Vega \cite[Theorem 2]{Vega1992}}]
  \label{vega}
  Let $d\in\N\setminus\{1\}$ and
  \begin{align*}
    1\leq p < \frac{2d}{d+1} \quad \text{and}\quad \frac12 \geq \frac{1}{\sigma'} \geq \max\left\{\frac{1}{p'},\frac{1}{p'}\ \frac{2d}{d-2}-\frac12\right\}\,.
  \end{align*}
  Then the restriction estimate
  \begin{align}
    \label{eq:vegarescaled}
    \begin{split}
      k^{-\frac{d-1}{2}}\|\hat f\|_{L^2(k\bs^{d-1})}
      = \|\hat f(k\cdot)\|_{L^2(\bs^{d-1})}
      & \lesssim \left[\int_0^\infty \left(\int_{\bs^{d-1}}|k^{-d}f(r\omega/k)|^\sigma\,\rd\omega\right)^{\frac p\sigma}r^{d-1}\,\dr\right]^{\frac1p}\\
      & = k^{-d/p'}\|f\|_{L^p(\R_+,r^{d-1}\,\dr:L^\sigma(\bs^{d-1}))}\,,
    \end{split}
  \end{align}
  the extension estimate
  \begin{align*}
    \|\F_{k \bs^{d-1}}^*g\|_{L^{p'}(\R_+,r^{d-1}\,\dr: L^{\sigma'}(\bs^{d-1}))}
    \lesssim k^{\frac{d-1}{2}-\frac{d}{p'}}\|g\|_{L^2(k\bs^{d-1})}\,,
  \end{align*}
  and the combined estimate
  \begin{align}
    \label{eq:vega3}
    \|\F_{k \bs^{d-1}}^* \F_{k\bs^{d-1}}\psi\|_{L^{p'}(\R_+,r^{d-1}\,\dr: L^{\sigma'}(\bs^{d-1}))}
    \lesssim k^{d-1-\frac{2d}{p'}}\|\psi\|_{L^p(\R_+,r^{d-1}\,\dr:L^\sigma(\bs^{d-1}))}
  \end{align}
  hold for all $k>0$ and are equivalent to each other.
\end{theorem}

Recall that the exponent $2d/(d+1)$ is sharp, i.e., Theorem \ref{vega}
cannot hold for $p\geq2d/(d+1)$.
Moreover, observe that the estimates are uniform in the radius $k$ as long
as $k\in[1-\delta,1+\delta]$ for some $0<\delta<1$. Theorem \ref{vega}
allows us to prove the following bound on the Birman--Schwinger operator.

\begin{proposition}
  \label{vegarestriction}
  Let $d\geq3$,
  $1\leq p<2d/(d+1)$, $1/2\geq 1/\sigma'\geq\max\{1/p',2d/(p'(d-2))-1/2\}$,
  $e>0$, and assume $T(\xi)=|\xi^2-1|$.
  Suppose $V \in L^{p/(2-p)}(\R_+,r^{d-1}\,\dr:\ L^{\sigma/(2-\sigma)}(\bs^{d-1}))$
  and, if $p\leq 2d/(d+2)$, suppose additionally $V\in L^{d/2}(\R^d)$.
  Then
  \begin{align*}
    \|(T+e)^{-1/2}|V|^{1/2}\|^2
    & \lesssim g(e) \left[\int_0^\infty \|V(r\cdot)\|_{L^{\sigma/(2-\sigma)}(\bs^{d-1})}^{p/(2-p)} r^{d-1}\,\dr\right]^{\!(2-p)/p} \\
    & \quad + \|V\|_{d/2}\theta(2d/(d+2)-p)
  \end{align*}
  where
  \begin{align}
    \label{eq:defg}
    g(e) = \int_{1/2}^{3/2} \frac{k^{d\left(\frac1p-\frac{1}{p'}\right)-1}}{T(k)+e}\,\dk \lesssim \max\{\ln(1/e),1\}\,.
  \end{align}
\end{proposition}

Here, $\theta$ denotes the Heaviside function with the convention $\theta(0)=1$.
Taking, e.g., $\sigma=2$ and $p\to 2d/(d+1)$ shows that spherically symmetric
$L^d(\R^d)$ potentials are almost admissible.

\begin{proof}
  Let $f=|V|^{1/2}\phi$ and $\phi\in L^2(\R^d)$.
  By H\"older's inequality, we have
  \begin{align*}
    \|f\|_{L^p(\R_+,r^{d-1}\,\dr:L^\sigma(\bs^{d-1}))} 
    \leq \left[\int_0^\infty \left(\int_{\bs^{d-1}}|V(r\omega)|^{\frac{\sigma}{2-\sigma}}\,\rd\omega\right)^{\frac{p(2-\sigma)}{\sigma(2-p)}}r^{d-1}\,\dr\right]^{\frac{2-p}{2p}}\|\phi\|_2\,.
  \end{align*}
  As in Section \ref{proofmainresult}, we consider small and large momenta
  separately and start with the latter.
  So let $\chi\in C_c^\infty(\R^d:[0,1])$ be a radial bump function centered
  at $|\xi|=1$ with $\supp\chi\subseteq \{\xi\in\R^d:1/2\leq|\xi|\leq3/2\}$.
  If $p\leq 2d/(d+2)$, the $L^{\frac{2d}{d+2}}\to L^2$ boundedness of
  $(T+e)^{-1/2}(1-\chi(-i\nabla))$ follows from
  $(T(\xi)+e)^{-1/2}(1-\chi)\lesssim (1+\xi^2)^{-1/2}$ and Sobolev
  embedding. Else, if $p>2d/(d+2)$, the $L^p\to L^2$ boundedness
  follows from Sobolev embedding, interpolation, and the fact that
  \begin{align*}
    \|f\|_p \lesssim \|f\|_{L^p(\R_+,r^{d-1}\,\dr:L^\sigma(\bs^{d-1}))} 
  \end{align*}
  by H\"older's inequality, since $\sigma\geq p$.

  Thus, we are left to estimate $\|\chi(T+e)^{-1/2}|V|^{1/2}\|_{2\to2}$
  with $p\in[1,2d/(d+1))$.
  By Plancherel, using spherical coordinates, and 
  Vega's estimate \eqref{eq:vegarescaled}, we obtain
  \begin{align*}
    \|\chi\, (T+e)^{-1/2}|V|^{1/2}\phi\|_2^2
    & = \int_0^\infty \dk\, \frac{\chi(k)^2 k^{d-1}}{T(k)+e} \int_{\bs^{d-1}} |\hat f(k\omega)|^2\,\rd\omega\\
    & \lesssim \|f\|_{L^p(\R_+,r^{d-1}\,\dr:L^\sigma(\bs^{d-1}))}^2 \int_{1/2}^{3/2} \dk\, \frac{k^{d-1-2d/p'}}{T(k)+e}\\
    & \leq g(e)\|V\|_{L^\frac{p}{2-p}(\R_+,r^{d-1}\,\dr:L^{\sigma/(2-\sigma)}(\bs^{d-1}))} \|\phi\|_2^2\,,
  \end{align*}
  which concludes the proof.
\end{proof}

\subsection{Estimates for potentials satisfying the MT condition}
We finally discuss potentials $V$ satisfying the
``radial Mizohata--Takeuchi'' condition.

\begin{definition}
  Let $V$ be a measurable, non-negative function on $\R^d$
  and $H(r) := \sup_{\omega\in\bs^{d-1}}V(r\omega)$. Then $V$ is said
  to satisfy the \emph{radial Mizohata--Takeuchi (MT) condition} if
  \begin{align}
    \label{eq:defmt}
    \|V\|_{\mt} := \sup_{\mu\geq0}\int_\mu^\infty\frac{H(r)r}{(r^2-\mu^2)^{1/2}}\,\dr<\infty\,.
  \end{align}
\end{definition}

Observe that $\|V(\cdot/k)\|_\mt = k\|V\|_\mt$ for all $k>0$.
We mention some examples of $V$ satisfying this condition.

\begin{example}~
  \begin{enumerate}
  \item Frank and Simon \cite[(4.2)]{MR3713021} showed
    $\|V\|_{\mt}\lesssim \|V\|_{L^{d,1}(\R_+,r^{d-1}\,\dr:\ L^\infty(\bs^{d-1}))}$,
    where
    $$
    \|V\|_{L^{d,1}(\R_+,r^{d-1}\,\dr:\ L^\infty(\bs^{d-1}))}
    := \int_0^\infty |\{r>0: \mathrm{ess-sup}_{\omega\in\bs^{d-1}}|V(r\omega)|>\alpha\}|_d^{1/d}\,\rd\alpha
    $$
    and $|\cdot|_d$ denotes the measure $|\bs^{d-1}|r^{d-1}\,\dr$.
  \item Barcelo, Ruiz, and Vega \cite[Proposition 1]{MR1479544} showed that for
    radial $V(x)=V(|x|)\equiv V(r)$, one has $\|V\|_{\mt}\lesssim \|V\|_{D_p}$
    for $p>2$, where
    $$
    \|V\|_{D_p} := \sum_{j=-\infty}^\infty \left(\int_{2^j}^{2^{j+1}}|V(r)|^p r^{p-1}\,\dr\right)^{1/p}\,.
    $$
    In particular, the functions $r^{-a}\one_{(0,1)}(r)+r^{-b}\one_{[1,\infty)}(r)$
    and $r^{-1}(1+|\log r|)^{-b}$ for $a<1$ and $b>1$ have finite
    $\|\cdot\|_{D_p}$ norm.
  \end{enumerate}
\end{example}

Barcelo, Ruiz, and Vega \cite{MR1479544} proved the following weighted
analog of the classical trace lemma.
\begin{theorem}[{\cite[Theorem 3]{MR1479544}}]
  \label{mt}
  Let $d\in\N\setminus\{1\}$ and $V$ be a radial, non-negative function
  satisfying $\|V\|_\mt<\infty$. Then the weighted restriction theorem
  \begin{align}
    \label{eq:barcelo}
    k^{-\frac{d-1}{2}} \|\hat f\|_{L^2(k\bs^{d-1})}
    = \|\hat f(k\cdot)\|_{L^2(\bs^{d-1})}
    \lesssim k^{-\frac{d-1}{2}}\|V\|_{\mt}^{1/2} \|f\|_{L^2(\R^d:V^{-1}(x)\,\dx)}\,,
  \end{align}
  the weighted extension estimate
  \begin{align*}
    \|\F_{k\bs^{d-1}}^*g\|_{L^2(V)} \lesssim \|V\|_\mt^{1/2} \|g\|_{L^2(k\bs^{d-1})}\,,
  \end{align*}
  and the combined estimate
  \begin{align}
    \label{eq:barcelo3}
    \|\F_{k\bs^{d-1}}^*\F_{k\bs^{d-1}}\psi\|_{L^2(V)}
    \lesssim \|V\|_\mt \|\psi\|_{L^2(V^{-1})}
  \end{align}
  hold for all $k>0$ and are equivalent to each other.
  Conversely, if one of the above estimates holds, and $V$ is radial
  and non-negative, then $\|V\|_\mt<\infty$.
\end{theorem}

\begin{remark}
  Using a result of Agmon and H\"ormander
  \cite[Theorem 3.1]{AgmonHormander1976}, Barcelo, Ruiz, and Vega
  \cite[p. 360-361]{MR1479544} observed that $V\in C_0(\R^d)$
  (vanishing at infinity, but not necessarily radial) has a uniformly
  bounded X-ray transform, i.e.,
  \begin{align}
    \label{eq:bddxray}
    \sup\left\{\int_\R V(y+t\omega)\,\dt:\ \omega\in\bs^{d-1}\,,y\in\R^d\right\} < \infty\,,
  \end{align}
  whenever the restriction estimate \eqref{eq:barcelo} holds for such $V$.
  By adapting their arguments to radial potentials $V$, they further remark
  that in this case \eqref{eq:bddxray} reduces to the MT condition
  \eqref{eq:defmt}.
  We shall, however, not make use of this remarkable fact in the following.
\end{remark}

Theorem \ref{mtrestriction} enables us to prove the following
Birman--Schwinger bound.

\begin{proposition}
  \label{mtrestriction}
  Assume $d\in\N\setminus\{1\}$, $e>0$, $T(\xi)=|\xi^2-1|$,
  and $V$ is a radial, non-negative function that satisfies
  the MT condition \eqref{eq:defmt}. Then
  \begin{align*}
    \|(T+e)^{-1/2}V^{1/2}\|^2 \lesssim (1+g_\mt(e)) \|V\|_{\mt}\,,
  \end{align*}
  where
  \begin{align}
    \label{eq:defgmt}
    g_\mt(e) = \int_{1/2}^{3/2} \frac{1}{T(k)+e}\,\dk \lesssim \max\{\ln(1/e),1\}\,.
  \end{align}
\end{proposition}

\begin{proof}
  Let $f=V^{1/2}\phi$ and $\phi\in L^2(\R^d)$, i.e.,
  $\|f\|_{L^2(V^{-1})}=\|\phi\|_2$ and $\chi$ be the same radial bump
  function in Fourier space as in the proof of Proposition
  \ref{vegarestriction}. First observe that large momenta are
  controlled by
  \begin{align*}
    \|(1-\chi)^2(T+e)^{-s/2}f\|_{L^2(V)}
    \lesssim \|(1-\Delta)^{-s/2}f\|_{L^2(V)}
    \leq \|V\|_\mt \|f\|_{L^2(V^{-1})}\,, \quad s\geq1
  \end{align*}
  where the second estimate is the content of \cite[Lemma 4]{MR1479544}.
  Taking $s=2$ and plugging in $f=V^{1/2}\phi$ shows
  \begin{align*}
    \|(1-\chi)(T+e)^{-1/2}V^{1/2}\|_{L^2\to L^2}^2
    = \|V^{1/2}(1-\chi)^2(T+e)^{-1}V^{1/2}\|_{L^2\to L^2}
    \lesssim \|V\|_\mt
  \end{align*}
  as desired.
  
  For momenta close to $\bs^{d-1}$, we use
  the restriction estimate \eqref{eq:barcelo} to obtain
  \begin{align*}
    \|\chi\,(T+e)^{-1/2}f\|_2^2
    & = \int_{0}^{\infty} \dk\, \frac{\chi(k)^2 k^{d-1}}{T(k)+e} \int_{\bs^{d-1}} |\hat f(k\omega)|^2\,\rd\omega\\
    & \lesssim \|V\|_\mt\|f\|_{L^2(V^{-1})}^2 \int_{1/2}^{3/2} \dk\, \frac{1}{T(k)+e}\\
    & = g_\mt(e)\|V\|_\mt \|\phi\|_{2}^2\,.
  \end{align*}
  This concludes the proof.
\end{proof}

\subsection{Weak coupling asymptotics for potentials in mixed norm spaces or satisfying the MT condition}

We are now in position to combine the above results in the following theorem.

\begin{theorem}
  \label{asymptoticsradial}
  Let $T(\xi)=|\xi^2-1|$, and suppose $d$ and $V$ satisfy the assumptions
  in Proposition \ref{vegarestriction} or \ref{mtrestriction}.
  Then for every eigenvalue $a_S^j>0$ of $\cv_S$ in \eqref{eq:lswbs2},
  counting multiplicity, and every $\lambda>0$, there is an eigenvalue
  $-e_j(\lambda)$ of $T(-i\nabla)-\lambda V$ with weak coupling limit
  \begin{align*}
    e_j(\lambda) = \exp\left(-\frac{1}{\lambda a_S^j}(1+o(1))\right)
    \quad \text{as}\ \lambda\to0\,.
  \end{align*}  
\end{theorem}

The notation $X\lesssim_V 1$ in the proof below conceals the more precise
estimate
$X\lesssim \min\{\|V\|_\mt,\|V\|_{L^{p/(2-p)}(\R_+,r^{d-1}\dr: L^{\sigma/(2-\sigma)}(\bs^{d-1}))}+\|V\|_{d/2}\theta(2d/(d+2)-p)\}$
for $p$ and $\sigma$ as in Proposition \ref{vegarestriction}.

\begin{proof}
  We follow the proof of Theorem \ref{thm. asymptotics}.
  First, the corresponding analog of Proposition \ref{proposition def. V_Sw},
  i.e., $L^2$ boundedness of $\cv_S$,
  follows immediately from the restriction theorems \ref{vega} and \ref{mt}
  above.
  Next, the splitting of $BS(e)$ is the same as in Section
  \ref{proofmainresult}. For concreteness, suppose that the frequency cutoff
  is at some $\tau<1/2$. As we have seen in the proofs of Propositions
  \ref{vegarestriction} and \ref{mtrestriction}, the high frequencies are
  harmless, i.e., in both cases we have $\|BS^{\rm high}(e)\|\lesssim_{\tau,V}1$,
  whereas the low frequencies are responsible for
  $\|BS^{\mathrm{low}}(e)\| \lesssim_{\tau,V} \ln(1/e)$ for $e\in(0,1)$.
  Thus, we are left to prove the analog of the key bound
  \eqref{Key estimate}.
  Recall that the Fermi surface of $T$ at energy $t\in(0,\tau]$ consists of
  the two connected components $S_t^\pm=\sqrt{1\pm t}\bs^{d-1}$.
  As explained below \eqref{BSlowreg} and in the proof of Lemma
  \ref{proofkeybound}, this merely relies on the validity of
  $\||V|^{1/2}\F_{S_t}^*\F_{S_t}V^{1/2}\|\lesssim_{\tau,V}1$ for all
  $t\in[0,\tau]$. But this is just reflected in \eqref{eq:vega3} and
  \eqref{eq:barcelo3} for $1/2<k<3/2$. This concludes the proof.
\end{proof}

\section*{Conflict of interest}

On behalf of all authors, the corresponding author states that
there is no conflict of interest.

\bibliographystyle{abbrv}

\end{document}